\documentclass[a4paper,onecolumn,superscriptaddress,11pt,accepted=2018-09-21]{quantumarticle}

\usepackage{tikzfig}
\usepackage{keycommand}
\usepackage{tikz}

\usetikzlibrary{shapes.multipart}

\tikzstyle{arrowhead}=[regular polygon,regular polygon sides=3,draw,scale=0.2,inner sep=-0.15pt,minimum width=6mm,fill=black,regular polygon rotate=180]
\tikzstyle{trace}=[circuit ee IEC,thick,ground,rotate=0,scale=2]
\tikzstyle{wavy}=[decorate,decoration={snake, segment length=1mm, amplitude=0.3mm}]
\tikzstyle{mopoint}=[shape=semicircle, fill=white,draw=black,shape border rotate=180,scale =0.75]
\tikzstyle{mocopoint}=[shape=semicircle, fill=white,draw=black,minimum width = 0.9cm, scale =0.75, xscale=0.7]
\tikzstyle{cpoint}=[shape=semicircle, fill=white,draw=black,minimum width = 0.9cm, scale =0.75, xscale=1, yscale=0.7, shape border rotate = 90,font=\fontsize{14}{16}\selectfont]
\tikzstyle{cocpoint}=[shape=semicircle, fill=white,draw=black,minimum width = 0.9cm, scale =0.75, xscale=1, yscale=0.7, shape border rotate = 270,font=\fontsize{14}{16}\selectfont]
\tikzstyle{slit}=[line width=2]
\tikzstyle{block}=[line width=4,red,line cap=round]
\tikzstyle{screen}=[line width=4,black,line cap=round]
\tikzstyle{di}=[diamond,draw,inner sep=0.5pt,font=\small, minimum size = .5cm]
\tikzstyle{sbox}=[rectangle,draw]
\tikzstyle{mirror}=[line width=2,black]
\tikzstyle{traceState}=[circuit ee IEC,thick,ground,rotate=180,scale=2]
\tikzstyle{detEff}=[circuit ee IEC,thick,ground,rotate=90,scale=1.4]
\tikzstyle{maxMix}=[circuit ee IEC,thick,ground,scale=1.4]
\tikzstyle{particlePath}=[line width=2,gray!40, line cap =round]

\tikzstyle{bwSpider}=[
       rectangle split,
       rectangle split parts=2,
       rectangle split part fill={black,white},
 minimum size=3.6 mm, inner sep=-2mm, draw=black,scale=0.5,rounded corners=0.8 mm
       ]
 \tikzstyle{wbSpider}=[
       rectangle split,
       rectangle split parts=2,
       rectangle split part fill={white,black},
 minimum size=3.6 mm, inner sep=-2mm, draw=black,scale=0.5,rounded corners=0.8 mm
       ]
\tikzstyle{cWire}=[densely dotted, thick]
\tikzstyle{env}=[copoint,regular polygon rotate=0,minimum width=0.2cm, fill=black]

\tikzstyle{probs}=[shape=semicircle,fill=white,draw=black,shape border rotate=180,minimum width=1.2cm]

\tikzstyle{every picture}=[baseline=-0.25em,scale=0.5]
\tikzstyle{dotpic}=[] 
\tikzstyle{diredges}=[every to/.style={diredge}]
\tikzstyle{math matrix}=[matrix of math nodes,left delimiter=(,right delimiter=),inner sep=2pt,column sep=1em,row sep=0.5em,nodes={inner sep=0pt},text height=1.5ex, text depth=0.25ex]

\tikzstyle{inline text}=[text height=1.5ex, text depth=0.25ex,yshift=0.5mm]
\tikzstyle{label}=[font=\footnotesize,text height=1.5ex, text depth=0.25ex,yshift=0.5mm]
\tikzstyle{left label}=[label,anchor=east,xshift=1.mm]
\tikzstyle{right label}=[label,anchor=west,xshift=-1.mm]

\tikzstyle{braceedge}=[decorate,decoration={brace,amplitude=2mm,raise=-1mm}]
\tikzstyle{small braceedge}=[decorate,decoration={brace,amplitude=1mm,raise=-1mm}]

\tikzstyle{doubled}=[line width=1.6pt] 
\tikzstyle{boldedge}=[doubled,shorten <=-0.17mm,shorten >=-0.17mm]
\tikzstyle{boldedgegray}=[doubled,gray,shorten <=-0.17mm,shorten >=-0.17mm]
\tikzstyle{singleedgegray}=[gray]

\tikzstyle{semidoubled}=[line width=1.4pt] 
\tikzstyle{semiboldedgegray}=[semidoubled,gray,shorten <=-0.17mm,shorten >=-0.17mm]

\tikzstyle{boxedge}=[semiboldedgegray]

\tikzstyle{boldedgedashed}=[very thick,dashed,shorten <=-0.17mm,shorten >=-0.17mm]
\tikzstyle{vboldedgedashed}=[doubled,dashed,shorten <=-0.17mm,shorten >=-0.17mm]
\tikzstyle{left hook arrow}=[left hook-latex]
\tikzstyle{right hook arrow}=[right hook-latex]
\tikzstyle{sembracket}=[line width=0.5pt,shorten <=-0.07mm,shorten >=-0.07mm]

\tikzstyle{causal edge}=[->,thick,gray]
\tikzstyle{causal nondir}=[thick,gray]
\tikzstyle{timeline}=[thick,gray, dashed]

\tikzstyle{cedge}=[<->,thick,gray!70!white]

\tikzstyle{empty diagram}=[draw=gray!40!white,dashed,shape=rectangle,minimum width=1cm,minimum height=1cm]
\tikzstyle{empty diagram small}=[draw=gray!50!white,dashed,shape=rectangle,minimum width=0.6cm,minimum height=0.5cm]

\tikzstyle{dot}=[inner sep=0mm,minimum width=2mm,minimum height=2mm,draw,shape=circle]

\tikzstyle{leak}=[white dot, shape=regular polygon, minimum size=3.3 mm, regular polygon sides=3, outer sep=-0.2mm, regular polygon rotate=270]
\tikzstyle{proj}=[white dot, shape=regular polygon, minimum size=3.3 mm, regular polygon sides=4, outer sep=-0.2mm]
\tikzstyle{Vleak}=[white dot, shape=regular polygon, minimum size=3.3 mm, regular polygon sides=3, outer sep=-0.2mm, regular polygon rotate=90]
\tikzstyle{dleak}=[white dot, line width=1.6pt, shape=regular polygon, minimum size=3.3 mm, regular polygon sides=3, outer sep=-0.2mm, regular polygon rotate=270]

\tikzstyle{Wsquare}=[white dot, shape=regular polygon, rounded corners=0.8 mm, minimum size=3.3 mm, regular polygon sides=3, outer sep=-0.2mm]
\tikzstyle{Wsquareadj}=[white dot, shape=regular polygon, rounded corners=0.8 mm, minimum size=3.3 mm, regular polygon sides=3, outer sep=-0.2mm, regular polygon rotate=180]
\tikzstyle{ddot}=[inner sep=0mm, doubled, minimum width=2.5mm,minimum height=2.5mm,draw,shape=circle]

\tikzstyle{black dot}=[dot,fill=black]
\tikzstyle{white dot}=[dot,fill=white,,text depth=-0.2mm]
\tikzstyle{white Wsquare}=[Wsquare,fill=gray,,text depth=-0.2mm]
\tikzstyle{white Wsquareadj}=[Wsquareadj,fill=white,,text depth=-0.2mm]
\tikzstyle{green dot}=[white dot] 
\tikzstyle{gray dot}=[dot,fill=gray!40!white,,text depth=-0.2mm]
\tikzstyle{red dot}=[gray dot] 

\tikzstyle{black ddot}=[ddot,fill=black]
\tikzstyle{white ddot}=[ddot,fill=white]
\tikzstyle{gray ddot}=[ddot,fill=gray!40!white]

\tikzstyle{gray edge}=[gray!60!white]

\tikzstyle{small dot}=[inner sep=0.5mm,minimum width=0pt,minimum height=0pt,draw,shape=circle]

\tikzstyle{small black dot}=[small dot,fill=black]
\tikzstyle{small white dot}=[small dot,fill=white]
\tikzstyle{small gray dot}=[small dot,fill=gray!40!white]

\tikzstyle{causal dot}=[inner sep=0.4mm,minimum width=0pt,minimum height=0pt,draw=white,shape=circle,fill=gray!40!white]

\tikzstyle{phase dimensions}=[minimum size=5mm,font=\footnotesize,rectangle,rounded corners=2.5mm,inner sep=0.2mm,outer sep=-2mm]
\tikzstyle{dphase dimensions}=[minimum size=5mm,font=\footnotesize,rectangle,rounded corners=2.5mm,inner sep=0.2mm,outer sep=-2mm]

\tikzstyle{white phase dot}=[dot,fill=white,phase dimensions]
\tikzstyle{white phase ddot}=[ddot,fill=white,dphase dimensions]

\tikzstyle{white rect ddot}=[draw=black,fill=white,doubled,minimum size=5mm,font=\footnotesize,rectangle,rounded corners=2.5mm,inner sep=0.2mm]
\tikzstyle{gray rect ddot}=[draw=black,fill=gray!40!white,doubled,minimum size=6mm,font=\footnotesize,rectangle,rounded corners=3mm]

\tikzstyle{gray phase dot}=[dot,fill=gray!40!white,phase dimensions]
\tikzstyle{gray phase ddot}=[ddot,fill=gray!40!white,dphase dimensions]
\tikzstyle{grey phase dot}=[gray phase dot]
\tikzstyle{grey phase ddot}=[gray phase ddot]

\tikzstyle{small phase dimensions}=[minimum size=4mm,font=\tiny,rectangle,rounded corners=2mm,inner sep=0.2mm,outer sep=-2mm]
\tikzstyle{small dphase dimensions}=[minimum size=4mm,font=\tiny,rectangle,rounded corners=2mm,inner sep=0.2mm,outer sep=-2mm]

\tikzstyle{small gray phase dot}=[dot,fill=gray!40!white,small phase dimensions]
\tikzstyle{small gray phase ddot}=[ddot,fill=gray!40!white,small dphase dimensions]

\tikzstyle{small map}=[draw,shape=rectangle,minimum height=4mm,minimum width=4mm,fill=white]

\tikzstyle{cnot}=[fill=white,shape=circle,inner sep=-1.4pt]

\tikzstyle{asym hadamard}=[fill=white,draw,shape=NEbox,inner sep=0.6mm,font=\footnotesize,minimum height=4mm]
\tikzstyle{asym hadamard conj}=[fill=white,draw,shape=NWbox,inner sep=0.6mm,font=\footnotesize,minimum height=4mm]
\tikzstyle{asym hadamard dag}=[fill=white,draw,shape=SEbox,inner sep=0.6mm,font=\footnotesize,minimum height=4mm]

\tikzstyle{hadamard}=[fill=white,draw,inner sep=0.6mm,font=\footnotesize,minimum height=4mm,minimum width=4mm]
\tikzstyle{small hadamard}=[fill=white,draw,inner sep=0.6mm,minimum height=1.5mm,minimum width=1.5mm]
\tikzstyle{small hadamard rotate}=[small hadamard,rotate=45]
\tikzstyle{dhadamard}=[hadamard,doubled]
\tikzstyle{small dhadamard}=[small hadamard,doubled]
\tikzstyle{small dhadamard rotate}=[small hadamard rotate,doubled]
\tikzstyle{antipode}=[white dot,inner sep=0.3mm,font=\footnotesize]

\tikzstyle{scalar}=[diamond,draw,inner sep=0.5pt,font=\small]
\tikzstyle{dscalar}=[diamond,doubled, draw,inner sep=0.5pt,font=\small]

\tikzstyle{small box}=[rectangle,inline text,fill=white,draw,minimum height=5mm,yshift=-0.5mm,minimum width=5mm,font=\small]
\tikzstyle{small gray box}=[small box,fill=gray!30]
\tikzstyle{medium box}=[rectangle,inline text,fill=white,draw,minimum height=5mm,yshift=-0.5mm,minimum width=10mm,font=\small]
\tikzstyle{square box}=[small box] 
\tikzstyle{medium gray box}=[small box,fill=gray!30]
\tikzstyle{semilarge box}=[rectangle,inline text,fill=white,draw,minimum height=5mm,yshift=-0.5mm,minimum width=12.5mm,font=\small]
\tikzstyle{large box}=[rectangle,inline text,fill=white,draw,minimum height=5mm,yshift=-0.5mm,minimum width=15mm,font=\small]
\tikzstyle{large gray box}=[small box,fill=gray!30]

\tikzstyle{Bayes box}=[rectangle,fill=black,draw, minimum height=3mm, minimum width=3mm]

\tikzstyle{gray square point}=[small box,fill=gray!50]

\tikzstyle{dphase box white}=[dhadamard]
\tikzstyle{dphase box gray}=[dhadamard,fill=gray!50!white]
\tikzstyle{phase box white}=[hadamard]
\tikzstyle{phase box gray}=[hadamard,fill=gray!50!white]

\tikzstyle{point}=[regular polygon,regular polygon sides=3,draw,scale=0.75,inner sep=-0.5pt,minimum width=9mm,fill=white,regular polygon rotate=180]
\tikzstyle{point nosep}=[regular polygon,regular polygon sides=3,draw,scale=0.75,inner sep=-2pt,minimum width=9mm,fill=white,regular polygon rotate=180]
\tikzstyle{copoint}=[regular polygon,regular polygon sides=3,draw,scale=0.75,inner sep=-0.5pt,minimum width=9mm,fill=white]
\tikzstyle{dpoint}=[point,doubled]
\tikzstyle{dcopoint}=[copoint,doubled]

\tikzstyle{pointgrow}=[shape=cornerpoint,kpoint common,scale=0.75,inner sep=3pt]
\tikzstyle{pointgrow dag}=[shape=cornercopoint,kpoint common,scale=0.75,inner sep=3pt]

\tikzstyle{wide copoint}=[fill=white,draw,shape=isosceles triangle,shape border rotate=90,isosceles triangle stretches=true,inner sep=0pt,minimum width=1.5cm,minimum height=6.12mm]
\tikzstyle{wide point}=[fill=white,draw,shape=isosceles triangle,shape border rotate=-90,isosceles triangle stretches=true,inner sep=0pt,minimum width=1.5cm,minimum height=6.12mm,yshift=-0.0mm]
\tikzstyle{wide point plus}=[fill=white,draw,shape=isosceles triangle,shape border rotate=-90,isosceles triangle stretches=true,inner sep=0pt,minimum width=1.74cm,minimum height=7mm,yshift=-0.0mm]

\tikzstyle{wide dpoint}=[fill=white,doubled,draw,shape=isosceles triangle,shape border rotate=-90,isosceles triangle stretches=true,inner sep=0pt,minimum width=1.5cm,minimum height=6.12mm,yshift=-0.0mm]

\tikzstyle{tinypoint}=[regular polygon,regular polygon sides=3,draw,scale=0.55,inner sep=-0.15pt,minimum width=6mm,fill=white,regular polygon rotate=180]

\tikzstyle{white point}=[point]
\tikzstyle{white dpoint}=[dpoint]
\tikzstyle{green point}=[white point] 
\tikzstyle{white copoint}=[copoint]
\tikzstyle{gray point}=[point,fill=gray!40!white]
\tikzstyle{gray dpoint}=[gray point,doubled]
\tikzstyle{red point}=[gray point] 
\tikzstyle{gray copoint}=[copoint,fill=gray!40!white]
\tikzstyle{gray dcopoint}=[gray copoint,doubled]

\tikzstyle{white point guide}=[regular polygon,regular polygon sides=3,font=\scriptsize,draw,scale=0.65,inner sep=-0.5pt,minimum width=9mm,fill=white,regular polygon rotate=180]

\tikzstyle{black point}=[point,fill=black,font=\color{white}]
\tikzstyle{black copoint}=[copoint,fill=black,font=\color{white}]

\tikzstyle{tiny gray point}=[tinypoint,fill=gray!40!white]

\tikzstyle{diredge}=[->]
\tikzstyle{ddiredge}=[<->]
\tikzstyle{rdiredge}=[<-]
\tikzstyle{thickdiredge}=[->, very thick]
\tikzstyle{pointer edge}=[->,very thick,gray]
\tikzstyle{pointer edge part}=[very thick,gray]
\tikzstyle{dashed edge}=[dashed]
\tikzstyle{thick dashed edge}=[very thick,dashed]
\tikzstyle{thick gray dashed edge}=[thick dashed edge,gray!40]
\tikzstyle{thick map edge}=[very thick,|->]

\makeatletter
\newcommand{\boxshape}[3]{%
\pgfdeclareshape{#1}{
\inheritsavedanchors[from=rectangle] 
\inheritanchorborder[from=rectangle]
\inheritanchor[from=rectangle]{center}
\inheritanchor[from=rectangle]{north}
\inheritanchor[from=rectangle]{south}
\inheritanchor[from=rectangle]{west}
\inheritanchor[from=rectangle]{east}
\backgroundpath{
\southwest \pgf@xa=\pgf@x \pgf@ya=\pgf@y
\northeast \pgf@xb=\pgf@x \pgf@yb=\pgf@y

\@tempdima=#2
\@tempdimb=#3

\pgfpathmoveto{\pgfpoint{\pgf@xa - 5pt + \@tempdima}{\pgf@ya}}
\pgfpathlineto{\pgfpoint{\pgf@xa - 5pt - \@tempdima}{\pgf@yb}}
\pgfpathlineto{\pgfpoint{\pgf@xb + 5pt + \@tempdimb}{\pgf@yb}}
\pgfpathlineto{\pgfpoint{\pgf@xb + 5pt - \@tempdimb}{\pgf@ya}}
\pgfpathlineto{\pgfpoint{\pgf@xa - 5pt + \@tempdima}{\pgf@ya}}
\pgfpathclose
}
}}

\boxshape{NEbox}{0pt}{5pt}
\boxshape{SEbox}{0pt}{-5pt}
\boxshape{NWbox}{5pt}{0pt}
\boxshape{SWbox}{-5pt}{0pt}
\boxshape{EBox}{-3pt}{3pt}
\boxshape{WBox}{3pt}{-3pt}
\makeatother

\tikzstyle{cloud}=[shape=cloud,draw,minimum width=1.5cm,minimum height=1.5cm]

\tikzstyle{map}=[draw,shape=NEbox,inner sep=2pt,minimum height=6mm,fill=white]
\tikzstyle{dashedmap}=[draw,dashed,shape=NEbox,inner sep=2pt,minimum height=6mm,fill=white]
\tikzstyle{mapdag}=[draw,shape=SEbox,inner sep=2pt,minimum height=6mm,fill=white]
\tikzstyle{mapadj}=[draw,shape=SEbox,inner sep=2pt,minimum height=6mm,fill=white]
\tikzstyle{maptrans}=[draw,shape=SWbox,inner sep=2pt,minimum height=6mm,fill=white]
\tikzstyle{mapconj}=[draw,shape=NWbox,inner sep=2pt,minimum height=6mm,fill=white]

\tikzstyle{medium map}=[draw,shape=NEbox,inner sep=2pt,minimum height=6mm,fill=white,minimum width=7mm]
\tikzstyle{medium map dag}=[draw,shape=SEbox,inner sep=2pt,minimum height=6mm,fill=white,minimum width=7mm]
\tikzstyle{medium map adj}=[draw,shape=SEbox,inner sep=2pt,minimum height=6mm,fill=white,minimum width=7mm]
\tikzstyle{medium map trans}=[draw,shape=SWbox,inner sep=2pt,minimum height=6mm,fill=white,minimum width=7mm]
\tikzstyle{medium map conj}=[draw,shape=NWbox,inner sep=2pt,minimum height=6mm,fill=white,minimum width=7mm]
\tikzstyle{semilarge map}=[draw,shape=NEbox,inner sep=2pt,minimum height=6mm,fill=white,minimum width=9.5mm]
\tikzstyle{semilarge map trans}=[draw,shape=SWbox,inner sep=2pt,minimum height=6mm,fill=white,minimum width=9.5mm]
\tikzstyle{semilarge map adj}=[draw,shape=SEbox,inner sep=2pt,minimum height=6mm,fill=white,minimum width=9.5mm]
\tikzstyle{semilarge map dag}=[draw,shape=SEbox,inner sep=2pt,minimum height=6mm,fill=white,minimum width=9.5mm]
\tikzstyle{semilarge map conj}=[draw,shape=NWbox,inner sep=2pt,minimum height=6mm,fill=white,minimum width=9.5mm]
\tikzstyle{large map}=[draw,shape=NEbox,inner sep=2pt,minimum height=6mm,fill=white,minimum width=12mm]
\tikzstyle{large map conj}=[draw,shape=NWbox,inner sep=2pt,minimum height=6mm,fill=white,minimum width=12mm]
\tikzstyle{very large map}=[draw,shape=NEbox,inner sep=2pt,minimum height=6mm,fill=white,minimum width=17mm]

\tikzstyle{medium dmap}=[draw,doubled,shape=NEbox,inner sep=2pt,minimum height=6mm,fill=white,minimum width=7mm]
\tikzstyle{medium dmap dag}=[draw,doubled,shape=SEbox,inner sep=2pt,minimum height=6mm,fill=white,minimum width=7mm]
\tikzstyle{medium dmap adj}=[draw,doubled,shape=SEbox,inner sep=2pt,minimum height=6mm,fill=white,minimum width=7mm]
\tikzstyle{medium dmap trans}=[draw,doubled,shape=SWbox,inner sep=2pt,minimum height=6mm,fill=white,minimum width=7mm]
\tikzstyle{medium dmap conj}=[draw,doubled,shape=NWbox,inner sep=2pt,minimum height=6mm,fill=white,minimum width=7mm]
\tikzstyle{semilarge dmap}=[draw,doubled,shape=NEbox,inner sep=2pt,minimum height=6mm,fill=white,minimum width=9.5mm]
\tikzstyle{semilarge dmap trans}=[draw,doubled,shape=SWbox,inner sep=2pt,minimum height=6mm,fill=white,minimum width=9.5mm]
\tikzstyle{semilarge dmap adj}=[draw,doubled,shape=SEbox,inner sep=2pt,minimum height=6mm,fill=white,minimum width=9.5mm]
\tikzstyle{semilarge dmap dag}=[draw,doubled,shape=SEbox,inner sep=2pt,minimum height=6mm,fill=white,minimum width=9.5mm]
\tikzstyle{semilarge dmap conj}=[draw,doubled,shape=NWbox,inner sep=2pt,minimum height=6mm,fill=white,minimum width=9.5mm]
\tikzstyle{large dmap}=[draw,doubled,shape=NEbox,inner sep=2pt,minimum height=6mm,fill=white,minimum width=12mm]
\tikzstyle{large dmap conj}=[draw,doubled,shape=NWbox,inner sep=2pt,minimum height=6mm,fill=white,minimum width=12mm]
\tikzstyle{large dmap trans}=[draw,doubled,shape=SWbox,inner sep=2pt,minimum height=6mm,fill=white,minimum width=12mm]
\tikzstyle{large dmap adj}=[draw,doubled,shape=SEbox,inner sep=2pt,minimum height=6mm,fill=white,minimum width=12mm]
\tikzstyle{large dmap dag}=[draw,doubled,shape=SEbox,inner sep=2pt,minimum height=6mm,fill=white,minimum width=12mm]
\tikzstyle{very large dmap}=[draw,doubled,shape=NEbox,inner sep=2pt,minimum height=6mm,fill=white,minimum width=19.5mm]

\tikzstyle{muxbox}=[draw,shape=rectangle,minimum height=3mm,minimum width=3mm,fill=white]
\tikzstyle{dmuxbox}=[muxbox,doubled]

\tikzstyle{box}=[draw,shape=rectangle,inner sep=2pt,minimum height=6mm,minimum width=6mm,fill=white]
\tikzstyle{dbox}=[draw,doubled,shape=rectangle,inner sep=2pt,minimum height=6mm,minimum width=6mm,fill=white]
\tikzstyle{dmap}=[draw,doubled,shape=NEbox,inner sep=2pt,minimum height=6mm,fill=white]
\tikzstyle{dmapdag}=[draw,doubled,shape=SEbox,inner sep=2pt,minimum height=6mm,fill=white]
\tikzstyle{dmapadj}=[draw,doubled,shape=SEbox,inner sep=2pt,minimum height=6mm,fill=white]
\tikzstyle{dmaptrans}=[draw,doubled,shape=SWbox,inner sep=2pt,minimum height=6mm,fill=white]
\tikzstyle{dmapconj}=[draw,doubled,shape=NWbox,inner sep=2pt,minimum height=6mm,fill=white]

\tikzstyle{ddmap}=[draw,doubled,dashed,shape=NEbox,inner sep=2pt,minimum height=6mm,fill=white]
\tikzstyle{ddmapdag}=[draw,doubled,dashed,shape=SEbox,inner sep=2pt,minimum height=6mm,fill=white]
\tikzstyle{ddmapadj}=[draw,doubled,dashed,shape=SEbox,inner sep=2pt,minimum height=6mm,fill=white]
\tikzstyle{ddmaptrans}=[draw,doubled,dashed,shape=SWbox,inner sep=2pt,minimum height=6mm,fill=white]
\tikzstyle{ddmapconj}=[draw,doubled,dashed,shape=NWbox,inner sep=2pt,minimum height=6mm,fill=white]

\boxshape{sNEbox}{0pt}{3pt}
\boxshape{sSEbox}{0pt}{-3pt}
\boxshape{sNWbox}{3pt}{0pt}
\boxshape{sSWbox}{-3pt}{0pt}
\tikzstyle{smap}=[draw,shape=sNEbox,fill=white]
\tikzstyle{smapdag}=[draw,shape=sSEbox,fill=white]
\tikzstyle{smapadj}=[draw,shape=sSEbox,fill=white]
\tikzstyle{smaptrans}=[draw,shape=sSWbox,fill=white]
\tikzstyle{smapconj}=[draw,shape=sNWbox,fill=white]

\tikzstyle{dsmap}=[draw,dashed,shape=sNEbox,fill=white]
\tikzstyle{dsmapdag}=[draw,dashed,shape=sSEbox,fill=white]
\tikzstyle{dsmaptrans}=[draw,dashed,shape=sSWbox,fill=white]
\tikzstyle{dsmapconj}=[draw,dashed,shape=sNWbox,fill=white]

\boxshape{mNEbox}{0pt}{10pt}
\boxshape{mSEbox}{0pt}{-10pt}
\boxshape{mNWbox}{10pt}{0pt}
\boxshape{mSWbox}{-10pt}{0pt}
\tikzstyle{mmap}=[draw,shape=mNEbox]
\tikzstyle{mmapdag}=[draw,shape=mSEbox]
\tikzstyle{mmaptrans}=[draw,shape=mSWbox]
\tikzstyle{mmapconj}=[draw,shape=mNWbox]

\tikzstyle{mmapgray}=[draw,fill=gray!40!white,shape=mNEbox]
\tikzstyle{smapgray}=[draw,fill=gray!40!white,shape=sNEbox]

\makeatletter

\pgfdeclareshape{cornerpoint}{
\inheritsavedanchors[from=rectangle] 
\inheritanchorborder[from=rectangle]
\inheritanchor[from=rectangle]{center}
\inheritanchor[from=rectangle]{north}
\inheritanchor[from=rectangle]{south}
\inheritanchor[from=rectangle]{west}
\inheritanchor[from=rectangle]{east}
\backgroundpath{
\southwest \pgf@xa=\pgf@x \pgf@ya=\pgf@y
\northeast \pgf@xb=\pgf@x \pgf@yb=\pgf@y

\pgfmathsetmacro{\pgf@shorten@left}{\pgfkeysvalueof{/tikz/shorten left}}
\pgfmathsetmacro{\pgf@shorten@right}{\pgfkeysvalueof{/tikz/shorten right}}

\pgfpathmoveto{\pgfpoint{0.5 * (\pgf@xa + \pgf@xb)}{\pgf@ya - 5pt}}
\pgfpathlineto{\pgfpoint{\pgf@xa - 8pt + \pgf@shorten@left}{\pgf@yb - 1.5 * \pgf@shorten@left}}
\pgfpathlineto{\pgfpoint{\pgf@xa - 8pt + \pgf@shorten@left}{\pgf@yb}}
\pgfpathlineto{\pgfpoint{\pgf@xb + 8pt - \pgf@shorten@right}{\pgf@yb}}
\pgfpathlineto{\pgfpoint{\pgf@xb + 8pt - \pgf@shorten@right}{\pgf@yb - 1.5 * \pgf@shorten@right}}
\pgfpathclose
}
}

\pgfdeclareshape{cornercopoint}{
\inheritsavedanchors[from=rectangle] 
\inheritanchorborder[from=rectangle]
\inheritanchor[from=rectangle]{center}
\inheritanchor[from=rectangle]{north}
\inheritanchor[from=rectangle]{south}
\inheritanchor[from=rectangle]{west}
\inheritanchor[from=rectangle]{east}
\backgroundpath{
\southwest \pgf@xa=\pgf@x \pgf@ya=\pgf@y
\northeast \pgf@xb=\pgf@x \pgf@yb=\pgf@y

\pgfmathsetmacro{\pgf@shorten@left}{\pgfkeysvalueof{/tikz/shorten left}}
\pgfmathsetmacro{\pgf@shorten@right}{\pgfkeysvalueof{/tikz/shorten right}}

\pgfpathmoveto{\pgfpoint{0.5 * (\pgf@xa + \pgf@xb)}{\pgf@yb + 5pt}}
\pgfpathlineto{\pgfpoint{\pgf@xa - 8pt + \pgf@shorten@left}{\pgf@ya + 1.5 * \pgf@shorten@left}}
\pgfpathlineto{\pgfpoint{\pgf@xa - 8pt + \pgf@shorten@left}{\pgf@ya}}
\pgfpathlineto{\pgfpoint{\pgf@xb + 8pt - \pgf@shorten@right}{\pgf@ya}}
\pgfpathlineto{\pgfpoint{\pgf@xb + 8pt - \pgf@shorten@right}{\pgf@ya + 1.5 * \pgf@shorten@right}}
\pgfpathclose
}
}

\makeatother

\pgfkeyssetvalue{/tikz/shorten left}{0pt}
\pgfkeyssetvalue{/tikz/shorten right}{0pt}

\tikzstyle{kpoint common}=[draw,fill=white,inner sep=1pt,minimum height=4mm]
\tikzstyle{kpoint sc}=[shape=cornerpoint,kpoint common]
\tikzstyle{kpoint adjoint sc}=[shape=cornercopoint,kpoint common]
\tikzstyle{kpoint}=[shape=cornerpoint,shorten left=5pt,kpoint common]
\tikzstyle{kpoint adjoint}=[shape=cornercopoint,shorten left=5pt,kpoint common]
\tikzstyle{kpoint conjugate}=[shape=cornerpoint,shorten right=5pt,kpoint common]
\tikzstyle{kpoint transpose}=[shape=cornercopoint,shorten right=5pt,kpoint common]
\tikzstyle{kpoint symm}=[shape=cornerpoint,shorten left=5pt,shorten right=5pt,kpoint common]

\tikzstyle{wide kpoint sc}=[shape=cornerpoint,kpoint common, minimum width=1 cm]
\tikzstyle{wide kpointdag sc}=[shape=cornercopoint,kpoint common, minimum width=1 cm]

\tikzstyle{black kpoint}=[shape=cornerpoint,shorten left=5pt,kpoint common,fill=black,font=\color{white}]

\tikzstyle{black kpoint sm}=[shape=cornerpoint,shorten left=5pt,kpoint common,fill=black,font=\color{white},scale=0.75]

\tikzstyle{black kpoint adjoint}=[shape=cornercopoint,shorten left=5pt,kpoint common,fill=black,font=\color{white}]
\tikzstyle{black kpointadj}=[shape=cornercopoint,shorten left=5pt,kpoint common,fill=black,font=\color{white}]

\tikzstyle{black kpointadj sm}=[shape=cornercopoint,shorten left=5pt,kpoint common,fill=black,font=\color{white},scale=0.75]

\tikzstyle{black dkpoint}=[shape=cornerpoint,shorten left=5pt,kpoint common,fill=black, doubled,font=\color{white}]
\tikzstyle{black dkpoint adjoint}=[shape=cornercopoint,shorten left=5pt,kpoint common,fill=black, doubled,font=\color{white}]
\tikzstyle{black dkpointadj}=[shape=cornercopoint,shorten left=5pt,kpoint common,fill=black, doubled,font=\color{white}]

\tikzstyle{black dkpoint sm}=[shape=cornerpoint,shorten left=5pt,kpoint common,fill=black, doubled,font=\color{white},scale=0.75]
\tikzstyle{black dkpointadj sm}=[shape=cornercopoint,shorten left=5pt,kpoint common,fill=black, doubled,font=\color{white},scale=0.75]

\tikzstyle{kpointdag}=[kpoint adjoint]
\tikzstyle{kpointadj}=[kpoint adjoint]
\tikzstyle{kpointconj}=[kpoint conjugate]
\tikzstyle{kpointtrans}=[kpoint transpose]

\tikzstyle{big kpoint}=[kpoint, minimum width=1.2 cm, minimum height=8mm, inner sep=4pt, text depth=3mm]

\tikzstyle{wide kpoint}=[kpoint, minimum width=1 cm, inner sep=2pt]
\tikzstyle{wide kpointdag}=[kpointdag, minimum width=1 cm, inner sep=2pt]
\tikzstyle{wide kpointconj}=[kpointconj, minimum width=1 cm, inner sep=2pt]
\tikzstyle{wide kpointtrans}=[kpointtrans, minimum width=1 cm, inner sep=2pt]

\tikzstyle{wider kpoint}=[kpoint, minimum width=1.25 cm, inner sep=2pt]
\tikzstyle{wider kpointdag}=[kpointdag, minimum width=1.25 cm, inner sep=2pt]
\tikzstyle{wider kpointconj}=[kpointconj, minimum width=1.25 cm, inner sep=2pt]
\tikzstyle{wider kpointtrans}=[kpointtrans, minimum width=1.25 cm, inner sep=2pt]

\tikzstyle{gray kpoint}=[kpoint,fill=gray!50!white]
\tikzstyle{gray kpointdag}=[kpointdag,fill=gray!50!white]
\tikzstyle{gray kpointadj}=[kpointadj,fill=gray!50!white]
\tikzstyle{gray kpointconj}=[kpointconj,fill=gray!50!white]
\tikzstyle{gray kpointtrans}=[kpointtrans,fill=gray!50!white]

\tikzstyle{gray dkpoint}=[kpoint,fill=gray!50!white,doubled]
\tikzstyle{gray dkpointdag}=[kpointdag,fill=gray!50!white,doubled]
\tikzstyle{gray dkpointadj}=[kpointadj,fill=gray!50!white,doubled]
\tikzstyle{gray dkpointconj}=[kpointconj,fill=gray!50!white,doubled]
\tikzstyle{gray dkpointtrans}=[kpointtrans,fill=gray!50!white,doubled]

\tikzstyle{white label}=[draw,fill=white,rectangle,inner sep=0.7 mm]
\tikzstyle{gray label}=[draw,fill=gray!50!white,rectangle,inner sep=0.7 mm]
\tikzstyle{black label}=[draw,fill=black,rectangle,inner sep=0.7 mm]

\tikzstyle{dkpoint}=[kpoint,doubled]
\tikzstyle{wide dkpoint}=[wide kpoint,doubled]
\tikzstyle{dkpointdag}=[kpoint adjoint,doubled]
\tikzstyle{wide dkpointdag}=[wide kpointdag,doubled]
\tikzstyle{dkcopoint}=[kpoint adjoint,doubled]
\tikzstyle{dkpointadj}=[kpoint adjoint,doubled]
\tikzstyle{dkpointconj}=[kpoint conjugate,doubled]
\tikzstyle{dkpointtrans}=[kpoint transpose,doubled]

\tikzstyle{kscalar}=[kpoint common, shape=EBox, inner xsep=-1pt, inner ysep=3pt,font=\small]
\tikzstyle{kscalarconj}=[kpoint common, shape=WBox, inner xsep=-1pt, inner ysep=3pt,font=\small]

\tikzstyle{spekpoint}=[kpoint sc,minimum height=5mm,inner sep=3pt]
\tikzstyle{spekcopoint}=[kpoint adjoint sc,minimum height=5mm,inner sep=3pt]

\tikzstyle{dspekpoint}=[spekpoint,doubled]
\tikzstyle{dspekcopoint}=[spekcopoint,doubled]

 \tikzstyle{upground}=[circuit ee IEC,thick,ground,rotate=90,scale=2.5]
 \tikzstyle{downground}=[circuit ee IEC,thick,ground,rotate=-90,scale=2.5]
 \tikzstyle{bigground}=[regular polygon,regular polygon sides=3,draw=gray,scale=0.50,inner sep=-0.5pt,minimum width=10mm,fill=gray]

\tikzstyle{arrs}=[-latex,font=\small,auto]
\tikzstyle{arrow plain}=[arrs]
\tikzstyle{arrow dashed}=[dashed,arrs]
\tikzstyle{arrow bold}=[very thick,arrs]
\tikzstyle{arrow hide}=[draw=white!0,-]
\tikzstyle{arrow reverse}=[latex-]
\tikzstyle{cdnode}=[]

\newcommand{\detEff}{\begin{tikzpicture}
	\begin{pgfonlayer}{nodelayer}
		\node [style=detEff] (0) at (0, 0.125) {};
		\node [style=none] (1) at (0, -0.25) {};
	\end{pgfonlayer}
	\begin{pgfonlayer}{edgelayer}
		\draw (0) to (1.center);
	\end{pgfonlayer}
\end{tikzpicture}}

\newkeycommand{\pointmap}[style=point][1]{\,\tikz{\node[style=\commandkey{style}] (x) at (0,-0.3) {$#1$};\node [style=none] (3) at (0, 1.0) {};\node [style=none] (3) at (0, -1.0) {};\draw (x)--(0,0.7);}\,}

\newkeycommand{\typedkpoint}[style=kpoint,edgestyle=][2]{%
\,\begin{tikzpicture}
  \begin{pgfonlayer}{nodelayer}
    \node [style=none] (0) at (0, 1) {};
    \node [style=\commandkey{style}] (1) at (0, -0.75) {$#2$};
    \node [style=right label] (2) at (0.25, 0.5) {$#1$};
  \end{pgfonlayer}
  \begin{pgfonlayer}{edgelayer}
    \draw [\commandkey{edgestyle}] (1) to (0.center);
  \end{pgfonlayer}
\end{tikzpicture}\,}

\newkeycommand{\typedkpointdag}[style=kpoint adjoint,edgestyle=][2]{%
\,\begin{tikzpicture}
  \begin{pgfonlayer}{nodelayer}
    \node [style=none] (0) at (0, -1) {};
    \node [style=\commandkey{style}] (1) at (0, 0.75) {$#2$};
    \node [style=right label] (2) at (0.25, -0.5) {$#1$};
  \end{pgfonlayer}
  \begin{pgfonlayer}{edgelayer}
    \draw [\commandkey{edgestyle}] (0.center) to (1);
  \end{pgfonlayer}
\end{tikzpicture}\,}

\newcommand{\bistateadj}[1]{\,\begin{tikzpicture}[yshift=-3mm]
  \begin{pgfonlayer}{nodelayer}
    \node [style=kpointadj, minimum width=1 cm, inner sep=2pt] (0) at (0, 0.5) {$\psi$};
    \node [style=none] (1) at (-0.75, 0.25) {};
    \node [style=none] (2) at (0.75, 0.25) {};
    \node [style=none] (3) at (-0.75, -0.5) {};
    \node [style=none] (4) at (0.75, -0.5) {};
  \end{pgfonlayer}
  \begin{pgfonlayer}{edgelayer}
    \draw (1.center) to (3.center);
    \draw (2.center) to (4.center);
  \end{pgfonlayer}
\end{tikzpicture}\,}

\newkeycommand{\pointketbra}[style1=copoint,style2=point][2]{\,%
\begin{tikzpicture}
  \begin{pgfonlayer}{nodelayer}
    \node [style=none] (0) at (0, -1.5) {};
    \node [style=\commandkey{style1}] (1) at (0, -0.75) {$#1$};
    \node [style=\commandkey{style2}] (2) at (0, 0.75) {$#2$};
    \node [style=none] (3) at (0, 1.5) {};
  \end{pgfonlayer}
  \begin{pgfonlayer}{edgelayer}
    \draw (0.center) to (1);
    \draw (2) to (3.center);
  \end{pgfonlayer}
\end{tikzpicture}\,}

\newkeycommand{\twocopointketbra}[style1=copoint,style2=copoint,style3=point][3]{\,%
\begin{tikzpicture}
  \begin{pgfonlayer}{nodelayer}
    \node [style=none] (0) at (-0.75, -1.5) {};
    \node [style=none] (1) at (0.75, -1.5) {};
    \node [style=\commandkey{style1}] (2) at (-0.75, -0.75) {$#1$};
    \node [style=\commandkey{style2}] (3) at (0.75, -0.75) {$#2$};
    \node [style=\commandkey{style3}] (4) at (0, 0.75) {$#3$};
    \node [style=none] (5) at (0, 1.5) {};
  \end{pgfonlayer}
  \begin{pgfonlayer}{edgelayer}
    \draw (0.center) to (2);
    \draw (1.center) to (3);
    \draw (4) to (5.center);
  \end{pgfonlayer}
\end{tikzpicture}\,}

\newkeycommand{\twopointketbra}[style1=copoint,style2=point,style3=point][3]{\,%
\begin{tikzpicture}
  \begin{pgfonlayer}{nodelayer}
    \node [style=none] (0) at (0, -1.5) {};
    \node [style=none] (1) at (-0.75, 1.5) {};
    \node [style=\commandkey{style1}] (2) at (0, -0.75) {$#1$};
    \node [style=\commandkey{style2}] (3) at (-0.75, 0.75) {$#2$};
    \node [style=\commandkey{style3}] (4) at (0.75, 0.75) {$#3$};
    \node [style=none] (5) at (0.75, 1.5) {};
  \end{pgfonlayer}
  \begin{pgfonlayer}{edgelayer}
    \draw (0.center) to (2);
    \draw (1.center) to (3);
    \draw (4) to (5.center);
  \end{pgfonlayer}
\end{tikzpicture}\,}

\newkeycommand{\pointbraket}[style1=point,style2=copoint][2]{\,%
\begin{tikzpicture}
  \begin{pgfonlayer}{nodelayer}
    \node [style=\commandkey{style1}] (0) at (0, -0.5) {$#1$};
    \node [style=\commandkey{style2}] (1) at (0, 0.5) {$#2$};
  \end{pgfonlayer}
  \begin{pgfonlayer}{edgelayer}
    \draw (0) to (1);
  \end{pgfonlayer}
\end{tikzpicture}\,}

\newkeycommand{\kpointbraket}[style1=kpoint,style2=kpoint adjoint][2]{\,%
\begin{tikzpicture}
  \begin{pgfonlayer}{nodelayer}
    \node [style=\commandkey{style1}] (0) at (0, -0.75) {$#1$};
    \node [style=\commandkey{style2}] (1) at (0, 0.75) {$#2$};
  \end{pgfonlayer}
  \begin{pgfonlayer}{edgelayer}
    \draw (0) to (1);
  \end{pgfonlayer}
\end{tikzpicture}\,}

\newkeycommand{\kpointmap}[style1=kpoint,style2=map][2]{\,%
\begin{tikzpicture}
  \begin{pgfonlayer}{nodelayer}
    \node [style=\commandkey{style1}] (0) at (0, -1.1) {$#1$};
    \node [style=\commandkey{style2}] (1) at (0, 0.2) {$#2$};
    \node [style=none] (2) at (0, 1.5) {};
  \end{pgfonlayer}
  \begin{pgfonlayer}{edgelayer}
    \draw (0) to (1);
    \draw (1) to (2.center);
  \end{pgfonlayer}
\end{tikzpicture}%
\,}

\newkeycommand{\sandwichtwo}[style1=point,style2=map,style3=map,style4=copoint][4]{\,%
\begin{tikzpicture}
  \begin{pgfonlayer}{nodelayer}
    \node [style=\commandkey{style2}] (0) at (0, -0.75) {$#2$};
    \node [style=\commandkey{style3}] (1) at (0, 0.75) {$#3$};
    \node [style=\commandkey{style1}] (2) at (0, -2) {$#1$};
    \node [style=\commandkey{style4}] (3) at (0, 2) {$#4$};
  \end{pgfonlayer}
  \begin{pgfonlayer}{edgelayer}
    \draw (2) to (0);
    \draw (0) to (1);
    \draw (1) to (3);
  \end{pgfonlayer}
\end{tikzpicture}\,}

\newkeycommand{\longbraket}[style1=point,style2=copoint][2]{\,%
\begin{tikzpicture}
  \begin{pgfonlayer}{nodelayer}
    \node [style=\commandkey{style1}] (0) at (0, -2) {$#1$};
    \node [style=\commandkey{style2}] (1) at (0, 2) {$#2$};
  \end{pgfonlayer}
  \begin{pgfonlayer}{edgelayer}
    \draw (0) to (1);
  \end{pgfonlayer}
\end{tikzpicture}\,}

\newkeycommand{\onb}[style=point]{\ensuremath{\left\{\tikz{\node[style=\commandkey{style}] (x) at (0,-0.3) {$j$};\draw (x)--(0,0.7);}\right\}}\xspace}

\newkeycommand{\copointmap}[style=copoint][1]{\,\tikz{\node[style=\commandkey{style}] (x) at (0,0.3) {$#1$};\draw (x)--(0,-0.7);}\,}

\tikzstyle{cdiag}=[matrix of math nodes, row sep=3em, column sep=3em, text height=1.5ex, text depth=0.25ex,inner sep=0.5em]
\tikzstyle{arrow above}=[transform canvas={yshift=0.5ex}]
\tikzstyle{arrow below}=[transform canvas={yshift=-0.5ex}]

\usepackage{geometry}
\usepackage[english]{babel}

\usepackage{hyperref}

\usepackage{palatino}
\usepackage{amsfonts}
\usepackage{amssymb}
\usepackage{amsmath}
\usepackage{amsthm}
\usepackage[numbers,sort&compress]{natbib}
\newcommand{\calP}{\mathcal{P}}
\newcommand{\calS}{\mathcal{S}}
\newcommand{\caltP}{\widetilde{\mathcal{P}}}
\newcommand{\wtalpha}{\widetilde{\alpha}}
\newcommand{\WNC}{\mathbf{WNC}}
\newcommand{\SNC}{\mathbf{SNC}}
\newcommand{\VS}{\mathbf{VS}}
\newcommand{\ket}[1]{| #1 \rangle}
\newcommand{\inner}[2]{\langle #1, #2 \rangle}
\newcommand{\enperiod}{\; .}
\newcommand{\encomma}{\; ,}

\newtheorem{defn}{Definition}
\newtheorem{thm}{Theorem}
\newtheorem{lm}{Lemma}
\newtheorem{remark}{Remark}
\newtheorem{fact}{Fact}
\newtheorem{cor}{Corollary}

\begin{document}

\title{How to make unforgeable money in\\ generalised probabilistic theories}
\date{\today}
\author{John H. Selby}
\affiliation{Perimeter Institute for Theoretical Physics, Waterloo, Ontario, Canada, N2L 2Y5}
\email{jselby@perimeterinstitute.ca}
\author{Jamie Sikora}
\affiliation{Perimeter Institute for Theoretical Physics, Waterloo, Ontario, Canada, N2L 2Y5}
\email{jsikora@perimeterinstitute.ca}

\maketitle

\begin{abstract}
We discuss the possibility of creating money that is physically impossible to counterfeit.
Of course, ``physically impossible'' is dependent on the theory that is a faithful description of nature.
Currently there are several proposals for quantum money which have their security based on the validity of quantum mechanics.
In this work, we examine Wiesner's money scheme in the framework of generalised probabilistic theories. This framework is broad enough to allow for essentially any potential theory of nature, provided that it admits an operational description.
We prove that under a quantifiable version of the no-cloning theorem, one can create physical money which has an exponentially small chance of being counterfeited. Our proof relies on cone programming, a natural generalisation of semidefinite programming. Moreover, we discuss some of the difficulties that arise when considering non-quantum theories.
\end{abstract}

\section{Introduction}
Since the discovery of quantum physics, there has been an ongoing effort to understand the technological impact it may have.
For example, it has been used to develop new technologies through a better control and understanding of microscopic systems, and, moreover, we are still trying to understand all of the information-theoretic advantages.
That is, we strive to better understand how the `weirdness' of the theory can be exploited for practical purposes.
There are countless examples found in the studies of quantum computation, information processing, and cryptography, and more are being discovered every day.
In this work, we focus on the important cryptographic task of creating money which is \emph{physically unforgeable}.

The money we use in our day to day lives only has value because it is difficult to counterfeit.
If we could easily duplicate it in some way then it would not take long before people were indeed taking advantage of this fact.
Indeed, despite the best government efforts, it was estimated that around $3\%$ of certain coins in the UK, for example, were counterfeits.
As it stands it is a constant battle between those that design the coins and those that try to counterfeit them.

There are numerous protocols for creating quantum money, \cite{Wie83, aaronson2012quantum, gavinsky2012quantum, MVW13} to name a few.
In fact, the very first cryptographic task using quantum information was a money scheme in Wiesner's seminal paper~\cite{Wie83}.
The key idea behind Wiesner's protocol, and those that followed, is that quantum theory could promise security based on the impossibility of cloning an unknown quantum state, and not on any technological limitations.
In other words, security based on the laws of physics rather than the limited resources of the counterfeiters.

This however is not the strongest form of security that one could imagine: it is contingent on our current best guess regarding the underlying physics which describes the world.
A more reliable form of security would be independent of a specific physical theory, and instead, be based on primitive physical principles that we may expect to hold regardless of the ultimate theory of nature.
The framework of Generalised Probabilistic Theories (GPTs) provides an operational framework in which we can address such problems.
For example the possibility of key distribution \cite{barrett2005no} and impossibility of bit commitment~\cite{chiribella2010probabilistic,SS17} have been demonstrated for a wide range of physical theories.

In this paper we explore the possibility of unforgeable money in the GPT framework.
Not only does this offer the potential for a much stronger foundation on which to base cryptographic security, but, it allows us to gain insight into quantum protocols by highlighting the key features of quantum theory necessary for security.

To prove our main result, we rely on the use of cone programming, also called linear conic optimisation. Cone programming is a generalisation of semidefinite programming which is another class of optimisation problems which has seen many uses in quantum theory.
The generalisation of semidefinite programming to cone programming mirrors that of quantum theory to GPTs.
For this reason, cone programming is a natural tool to have handy when studying post-quantum theories.
Although cone programming is a well-studied area of optimisation theory, it has only had a small number of applications in quantum theory~\cite{GSU13, BCJRWY14, LP15, NST16, SW17} and in GPTs~\cite{fiorini2014generalized, JP17, SS17,bae2016structure,LPW17}.
We hope this work will inspire future applications of cone programming in the study of GPTs and solidify it as an indispensable mathematical tool.

\section{Unforgeable money: the idea}

The idea behind the first quantum scheme for unforgeable money is quite simple: if a banknote contains unknown physical states, then the no-cloning theorem proves that the money cannot be duplicated.
However, even within quantum theory there are some caveats.
Due to the uncertainty principle, one cannot ascertain the exact state of a quantum system.
Thus, one could imagine a perfect copy is not needed to counterfeit money, only the ability to cheat someone who might be testing if counterfeiting occurred (with a reasonable probability of success).

In Wiesner's original scheme~\cite{Wie83}, the bank randomly selects one of the following four qubits
$\ket{0}, \ket{1}, \ket{+}, \ket{-}$
and embeds it into the banknote.
When the holder of the banknote wants to verify the banknote is authentic, they tell the bank the serial number, then the bank looks up what the state should be, then measures to see if the state is intact.
Indeed, there is a less-than-perfect chance of creating two banknotes which will each pass this verification.
Intuitively, the more qubits the bank puts into the banknote, the harder it is to counterfeit.
This intuition was later proven to be true in~\cite{MVW13} through the use of semidefinite programming.

In this paper, we wish to see if something like the above holds in GPTs.
Of course, one needs to define what the physical states are, and how the bank verifies them.
Without going into detail yet, we just assume the bank embeds a physical state into the banknote, and independently verifies each copy individually.
The bank's verification must be a physical process, and we desire it to be secure against all physically-realisable counterfeiting machines.
Roughly speaking, let $C$ represent the bank's entire strategy of creating and verifying a banknote, and let $\calP$ be the set of all physical counterfeiting machines (not necessarily perfect ones).
Then we would like to design a money scheme such that the following quantity is as small as possible:
\begin{equation}
\sup \{ C(X) : X \in \calP \}
\end{equation}
where $C(X)$ denotes the probability that the two copies of the counterfeiter each pass the bank's verification procedure.

We shortly introduce the physics required to establish meaningful definitions of $C$ and $\calP$.

\section{Background: Generalised probabilistic theories}

We now introduce the mathematics of the framework for generalised probabilistic theories that we use in this paper (those familiar with this topic can skim this section for notation and  proceed to Section~\ref{GPTmoney}).
For simplicity we present the mathematical bare bones of the framework, but note that this can be derived from the basic operational ideas regarding the classical interface for a theory \cite{selby2017process,selbyReconstruction}.
The framework we present here is closely related to many other approaches to generalised probabilistic theories (e.g.,  \cite{hardy2001quantum,barrett2007information,Ludwig,davies1970operational,randall1970approach,Piron64,Mackey}), but, in particular to the work of \cite{chiribella2010probabilistic} and \cite{hardy2011reformulating} which also take a diagrammatic framework as a foundation.

The starting point for our framework is the notion of a process theory \cite{coecke2017picturing,coecke2015categorical,selby2017process}.
The key feature of a process theory is that it provides a diagrammatic representation of the theory.
There are two primitive components of a process theory, physical \emph{systems}, denoted by labeled wires, and physical \emph{processes} which can have input and output systems, denoted as a labeled box with input wires at the bottom and output wires at the top.
These processes can be `wired together' to form diagrams, for example,
\begin{equation}
\InputIfFileExists{Diagrams/genericDiagram.tikz}{}{\input{./figures/Diagrams/genericDiagram.tikz}} \enperiod
\end{equation}
Such diagrams are themselves valid processes in the theory; in this case with a \emph{composite} system  $AC$ as an input and $BE$ as an output.
Simply put, processes are \emph{closed} under being wired together.
Processes with no inputs (such as $c$ in the above diagram) are called \emph{states}, those with no outputs (such as $b$ above) {are} called \emph{effects}, and those with neither are simply called \emph{numbers}.
Numbers can be, for example, obtained when composing  a state with an effect, and so, as we are interested in probabilistic theories, these numbers are taken to be the non-negative reals, $\mathbb{R}^+$.

{Equality of processes can then be characterised in terms of these numbers, the idea being that if two processes give the same probabilities in all situations then the are the equal. This defines the notion of \emph{tomography}, which formally can be expressed as $f,g:A\to B$ are equal if and only if
\begin{equation}\label{eq:Tomog}
\InputIfFileExists{Diagrams/tomographyf.tikz}{}{\input{./figures/Diagrams/tomographyf.tikz}}\ = \ %
\InputIfFileExists{Diagrams/tomographyg.tikz}{}{\input{./figures/Diagrams/tomographyg.tikz}} \quad \forall C, s, e.
\end{equation}
}

We consider process theories that come with a way to discard (or simply ignore) systems.
For this purpose we introduce a discarding effect for each system $A$ denoted as:
\begin{equation}
\begin{tikzpicture}
	\begin{pgfonlayer}{nodelayer}
		\node [style=none] (0) at (0, -0.75) {};
		\node [style=none] (1) at (0, 0.25) {};
		\node [style=upground] (2) at (0, 0.5) {};
		\node [style={right label}] (3) at (0, -0.5) {$A$};
	\end{pgfonlayer}
	\begin{pgfonlayer}{edgelayer}
		\draw (1.center) to (0.center);
	\end{pgfonlayer}
\end{tikzpicture}} \enperiod
\end{equation}
Moreover we require that the composite of discarding effects is the discarding effect for the composite system
\begin{equation}
\InputIfFileExists{Diagrams/discardEffectComposite.tikz}{}{\input{./figures/Diagrams/discardEffectComposite.tikz}} \enperiod
\end{equation}
In particular, discarding ``nothing'', i.e., the trivial system, is just the number $1$.

These discarding effects then allow us to define a notion of \emph{causality} for processes  \cite{chiribella2010probabilistic,coecke2014terminality,kissinger2017equivalence}. Specifically, a process $f$ is causal if and only if
\begin{equation} \label{causal}
{%
\InputIfFileExists{Diagrams/causalProcess.tikz}{}{\input{./figures/Diagrams/causalProcess.tikz}}} \enperiod
\end{equation}
The reason for naming these `causal' is not immediately apparent, however it can be shown that restricting to the causal processes of a theory ensures compatibility with the causal structure of relativity \cite{kissinger2017equivalence}.
For example, this ensures that there is no signalling back in time \cite{chiribella2010probabilistic} or faster than the speed of light \cite{coecke2014terminality}.

On the other hand, restricting to just the causal processes is a step too far. Indeed, the only causal number is $1$ and so such processes only describe deterministic situations.
We want to discuss probabilistic scenarios where the numbers correspond to the probability that some event occurs.
To deal with this we therefore also work with \emph{subcausal} processes, that is, processes which can occur as some probabilistic `branch' of a deterministic process.
To formalise this branching structure we introduce diagrammatic sums, i.e., a sum of processes that distributes over diagrams:
\begin{equation}\label{eq:sum}
\InputIfFileExists{Diagrams/distributivity.tikz}{}{\input{./figures/Diagrams/distributivity.tikz}}
\end{equation}
where $\xi$ is shorthand notation for an arbitrary diagram
of the form
\begin{equation}
\InputIfFileExists{Diagrams/comb2.tikz}{}{\input{./figures/Diagrams/comb2.tikz}} \enperiod
\end{equation}
An important consequence of this sum is that it allows us to define a partial order for processes:
\begin{equation}\label{eq:partialOrder}
\InputIfFileExists{Diagrams/partialOrder.tikz}{}{\input{./figures/Diagrams/partialOrder.tikz}}
\end{equation}
where $z$ is another process in the theory. In particular, given this partial order we can define the subcausal processes as those for which the following holds:
\begin{equation}\label{eq:subcausal}
{%
\InputIfFileExists{Diagrams/subCausal1.tikz}{}{\input{./figures/Diagrams/subCausal1.tikz}}} \enperiod
\end{equation}
Importantly, the set of subcausal processes is closed under composition and the subcausal numbers are given by the interval $[0,1]$ so this faithfully captures the probabilistic part of the process theory as we required.

In the definition of this ordering we have assumed that $z$ is a {physical} process.
Later we consider the case when $z$ is not a {physical} process, but belongs to some set $K$, and denote such an ordering by `$\leq_K$'.
Note that such an ordering may not have any physical meaning, but it will have a mathematical meaning once we set up the mathematical structure behind the sets of different processes.
It is understood that the `absence of explicitly mentioning such a set' means that $K$ is the set of processes.
For example, in Eq.~\eqref{eq:partialOrder}, `$\leq$' is shorthand for $K$ being the set of processes from $A$ to $B$ and in Eq.~\eqref{eq:subcausal} it is shorthand for $K$ being the set of effects on $B$.

Given this structure one can observe that the set of processes $\{ f_i \}$ with a given (potentially composite) input $A$ and (potentially composite) output $B$ has a rich linear structure.
Specifically, using the sum defined in Eq.~\eqref{eq:sum}, we can form (non-negative) linear combinations as
\begin{equation}\label{eq:linear}
\InputIfFileExists{Diagrams/linearCombination.tikz}{}{\input{./figures/Diagrams/linearCombination.tikz}} \encomma
\; \text{for } r_i\in\mathbb{R}^+
\end{equation}
which is itself a valid process. These processes therefore form a convex cone, denoted as $K_A^B$, which naturally extends to a vector space $V_A^B$ which is spanned by the cone.
Note that this also holds for states and effects as they are just special instances of processes, i.e., we obtain a cone of states $K^A$ (which have no input) and a cone of effects $K_A$ (which have no output).
The causal processes form a convex set inside this cone, that is, if $f$ and $g$ are causal then it is simple to check that
\begin{equation}
\InputIfFileExists{Diagrams/convexCombination.tikz}{}{\input{./figures/Diagrams/convexCombination.tikz}}
\end{equation}
is causal as well for any $p \in [0,1]$.
This convex set can be characterised by those processes in $K_A^B$ that satisfy Eq.~\eqref{causal}.
Importantly for our result, it can then be shown that there are causal processes in the \emph{interior} of the cone $K_A^B$.
Additionally, we will make the assumption that the cones are \emph{closed} (as in practice a theory should be operationally indistinguishable from its closure) and \emph{finite-dimensional} (as it is impossible in practice to do an infinite number of experiments to characterise processes).

We illustrate these concepts with three specific examples of processes, below.

\smallskip \noindent
\textbf{States:} In general this can be an arbitrary convex cone, the causal states are given by the intersection of this cone with a single hyperplane and the subcausal states lie between this hyperplane and the zero vector (i.e. the point), as depicted below.
\begin{equation}
\includegraphics[clip, trim=1cm 13.5cm 1cm 5cm,width=0.7\textwidth]{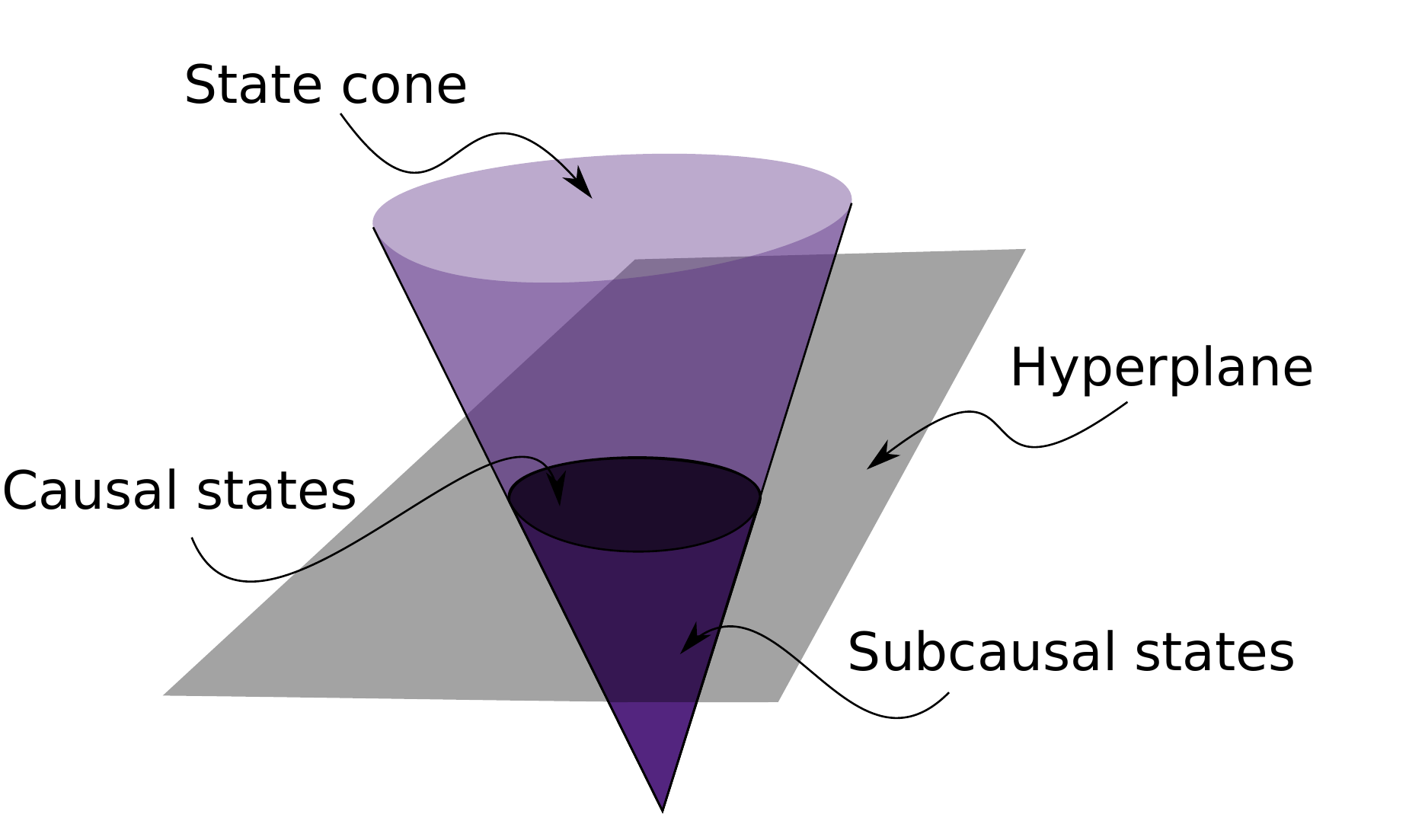}
\end{equation}

\smallskip \noindent
\textbf{Effects:}
We have only a single causal effect, the discarding effect itself. The subcausal effects lie between the discarding effect and the zero vector.
\begin{equation}
\includegraphics[clip, trim=1cm 17cm 1cm 4cm,width=0.7\textwidth]{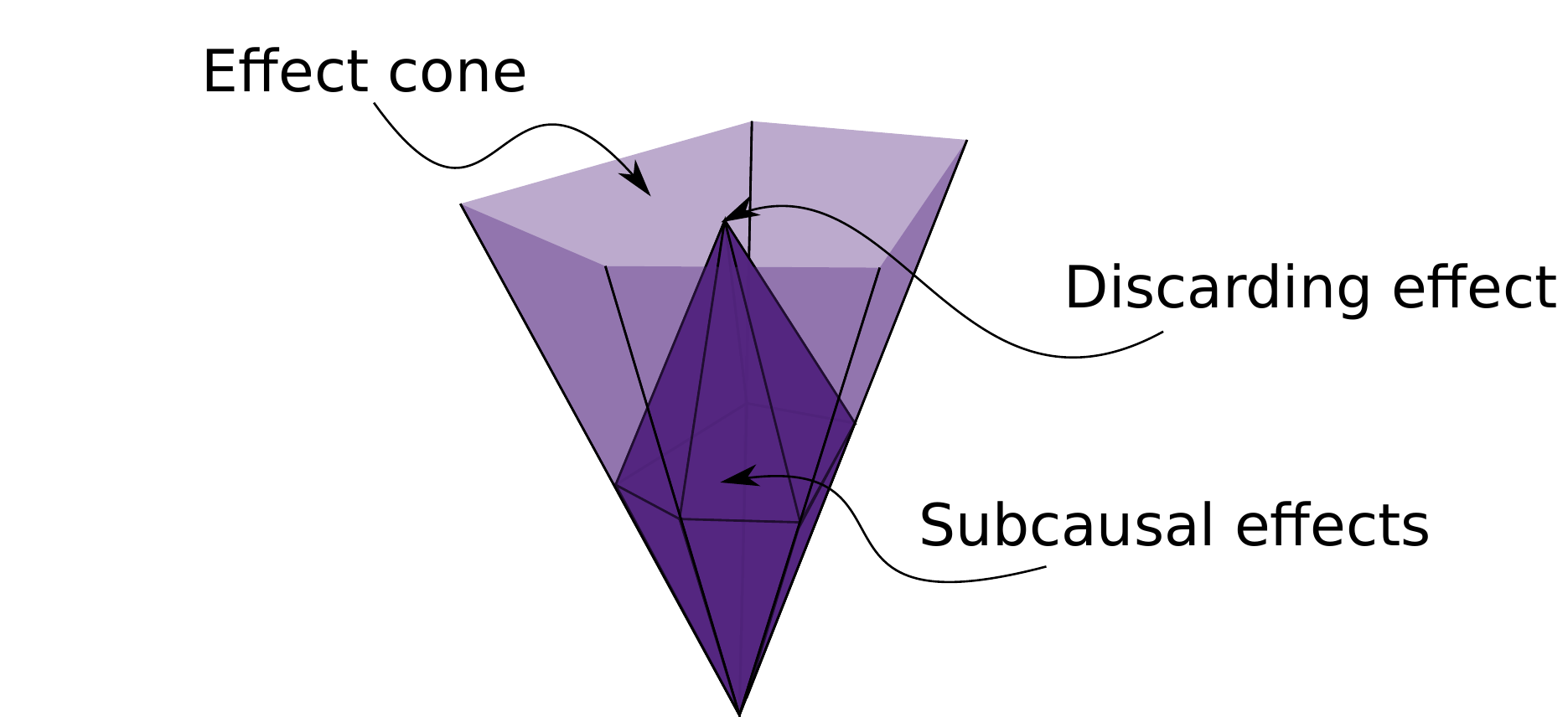}
\end{equation}

\smallskip \noindent
\textbf{Transformations:}
The set of causal processes is a more complicated set, but  importantly, it is still an \emph{affine} constraint, namely Eq.~\eqref{causal}, on the transformation cone. The subcausal transformations lie between this affine constraint and the zero vector.
\begin{equation}
\includegraphics[clip, trim=0cm 15.5cm 1cm 3.5cm,width=0.7\textwidth]{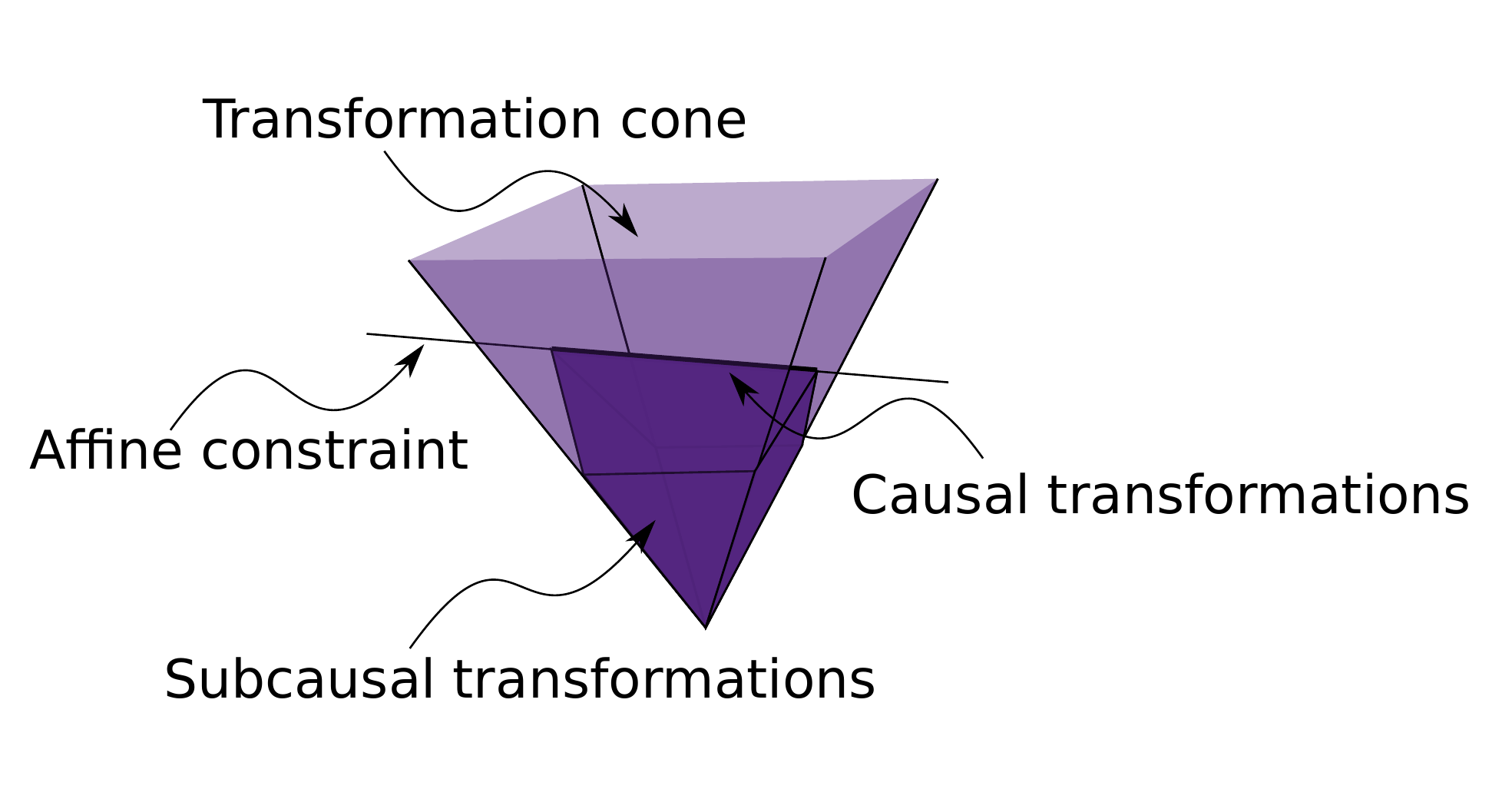}
\end{equation}

\newpage
\section{Background: Cone programming}
\label{convex}

{\color{black}
Cone programming is the study of optimization problems of the form
\begin{equation} \label{primal2}
\gamma = \sup \left\{ \inner{C}{X} : b - \phi(X) \in K, X \in K' \right\}
\end{equation}
where $K$ and $K'$ are finite-dimensional convex cones, $C$ and $b$ are vectors, $\phi$ is a linear function, and $X$ is the variable over which we wish to optimize.

Before continuing, we must introduce the notion of a \emph{dual cone}.
\begin{defn} \label{dualcone}
Given a convex cone $K$, we define its \emph{dual cone} as
\begin{equation}
K^* := \{ v : \inner{v}{s} \geq 0, \forall s\in K \}.
\end{equation}
\end{defn}

In this section, we study the cone program of the form below:
\begin{equation} \label{primal2}
\gamma = \sup \left\{ \inner{C}{X} : b - \phi(X) \in K_1^*, X \in K_2 \right\}
\end{equation}
where $K_1$ and $K_2$ are \emph{closed} convex cones.
The first step in understanding a cone program is to examine its dual cone program, which is another cone program related to the original. In this case, the dual is given as
\begin{equation} \label{dual2}
\beta = \inf \left\{ \inner{b}{y} : \phi^*(y) - C \in K_2^*, y \in K_1 \right\}
\end{equation}
where $\phi^*$ is the adjoint of $\phi$.

Suppose there exists an $X$ satisfying
\begin{equation} \label{primalslater}
X \in \mathrm{int}(K_2) \text{ and } b - \phi(X) \in \mathrm{int}(K_1^*)
\end{equation}
and a $y$ satisfying
\begin{equation} \label{dualslater}
y \in \mathrm{int}(K_1) \text{ and } \phi^*(y) - C \in \mathrm{int}(K_2^*).
\end{equation}
Then we have $\gamma = \beta$ and both problems attain an optimal solution.
This is known as \emph{strong duality}, the proof of which is beyond the scope of this section. We refer the interested reader to the book~\cite{BV} for a proof.

All of the results relying on cone programming in this paper only use strong duality. The dual yields a new perspective on studying a cone program while strong duality proves that it is very useful, especially when we want a new way to write the optimal value.
}
\section{GPT Money: the scheme and measuring its security} \label{GPTmoney}

We now describe the physical process of creating the banknote, its verification, and what counterfeiting machines a counterfeiter can use.

\medskip
\noindent
\textbf{Preparation:}
The bank selects a causal state $s_i$ of some system $A$ from the ensemble $\varepsilon_A = \{ (p_1, s_1), \ldots, (p_n, s_n) \}$ with $p_1, \ldots, p_n > 0$ and ${\sum_{i=1}^n p_i =1}$. Therefore, for all $i$ we have
\begin{equation}
\InputIfFileExists{Diagrams/causal_s_i.tikz}{}{\input{./figures/Diagrams/causal_s_i.tikz}}\ =\ 1 \enperiod
\end{equation}
The bank puts the state into the banknote and records its serial number as well as the state selected.

\medskip
\noindent
\textbf{Verification (of a single copy):}
The bank uses the effect $e_i$ to verify the state $s_i$ where $e_i$ is subcausal and always accepts $s_i$, i.e., for each $i$, we have
\begin{equation} \label{effect1}
\begin{tikzpicture}
	\begin{pgfonlayer}{nodelayer}
		\node [style=none] (0) at (0, 0) {};
		\node [style=none] (1) at (0, -0.5) {};
		\node [style=copoint] (2) at (0, 0.25) {$e_i$};
		\node [style=right label] (3) at (0, -0.5) {$A$};
	\end{pgfonlayer}
	\begin{pgfonlayer}{edgelayer}
		\draw (1.center) to (0.center);
	\end{pgfonlayer}
\end{tikzpicture}
}\ \leq\ %
\begin{tikzpicture}
	\begin{pgfonlayer}{nodelayer}
		\node [style=none] (0) at (0, -0.5) {};
		\node [style=none] (1) at (0, -0) {};
		\node [style=upground] (2) at (0, 0.25) {};
		\node [style={right label}] (3) at (0, -0.5) {$A$};
	\end{pgfonlayer}
	\begin{pgfonlayer}{edgelayer}
		\draw (1.center) to (0.center);
	\end{pgfonlayer}
\end{tikzpicture}
}
\end{equation}
and
\begin{equation} \label{effect2}
\InputIfFileExists{Diagrams/sharpNEW.tikz}{}{\input{./figures/Diagrams/sharpNEW.tikz}}\ =\ 1
\enperiod
\end{equation}
This way the verification always passes when the state is untampered.
Note that one could consider more general bank strategies, but, as we are interested in proving the impossibility of counterfeiting, allowing for more general strategies could only make the counterfeiters job more difficult.

\begin{defn}[Bank strategy]
For the ensemble $\varepsilon_A = \{ (p_1, s_1), \ldots, (p_n, s_n) \}$ and corresponding verification effects $\{ e_1, \ldots, e_n \}$, satisfying Eqs.~\eqref{effect1} and \eqref{effect2}, we call the set
\begin{equation}
\calS_A = \{ (p_1, s_1, e_1), \ldots, (p_n, s_n, e_n) \}
\end{equation}
a bank strategy.
Note that this fully describes what the bank does to prepare and verify a banknote.
\end{defn}

\medskip
\noindent
\textbf{Counterfeiting machines:}
We now flesh out the details of the set of physical counterfeiting machines $\calP$.
The counterfeiters strategy is simple to define, it is a physical process which takes in the state of system $A$ given by bank, and outputs some state of system $AA$ which is intended to be two copies of the original\footnote{{One could consider a more elaborate scenario in which the counterfeiter has access to $n$ banknotes and just aims to produce $m>n$ at the end, we however leave this and other such elaborations to future work.}}.
Mathematically,  it can be represented by some subcausal $\chi \in K_A^{AA}$, i.e.,
\begin{equation}
\InputIfFileExists{Diagrams/causalXLabeled.tikz}{}{\input{./figures/Diagrams/causalXLabeled.tikz}} \ \leq
 \ %
\begin{tikzpicture}
	\begin{pgfonlayer}{nodelayer}
		\node [style=none] (0) at (0, -0.75) {};
		\node [style=none] (1) at (0, 0.25) {};
		\node [style=upground] (2) at (0, 0.5) {};
		\node [style=right label] (3) at (0, -0.5) {$A$};
	\end{pgfonlayer}
	\begin{pgfonlayer}{edgelayer}
		\draw (1.center) to (0.center);
	\end{pgfonlayer}
\end{tikzpicture}
} \enperiod
\end{equation}
Thus, the set of physical counterfeiting machines is given as
\begin{equation}
\calP_A = \left\{\ \ %
\InputIfFileExists{Diagrams/causalXLabeled.tikz}{}{\input{./figures/Diagrams/causalXLabeled.tikz}} \ \leq
 \ %
\begin{tikzpicture}
	\begin{pgfonlayer}{nodelayer}
		\node [style=none] (0) at (0, -0.75) {};
		\node [style=none] (1) at (0, 0.25) {};
		\node [style=upground] (2) at (0, 0.5) {};
		\node [style=right label] (3) at (0, -0.5) {$A$};
	\end{pgfonlayer}
	\begin{pgfonlayer}{edgelayer}
		\draw (1.center) to (0.center);
	\end{pgfonlayer}
\end{tikzpicture}
}, \quad  %
\InputIfFileExists{Diagrams/counterfeiterChannel.tikz}{}{\input{./figures/Diagrams/counterfeiterChannel.tikz}}\in K_{A}^{AA} \ \right\}.
\end{equation}

\medskip
\noindent
\textbf{Security:}
Suppose the bank is independently given each output system and so independently tests each
of them with the relevant effect\footnote{{One could also consider more elaborate bank strategies, but our work here is a proof of principle and so we will leave these developments to future work.}}.
The overall bank verification process is therefore given by the following diagram:
\begin{equation} %
\InputIfFileExists{Diagrams/bankStrategyCLabeled.tikz}{}{\input{./figures/Diagrams/bankStrategyCLabeled.tikz}}\ :=\ %
\InputIfFileExists{Diagrams/bankStrategy.tikz}{}{\input{./figures/Diagrams/bankStrategy.tikz}}
\enperiod
\end{equation}
With these definitions in hand, the quantity we use to measure the security of the money scheme is given by
\begin{equation} \label{CP}
\alpha_A :=
\sup\left\{\  %
\InputIfFileExists{Diagrams/probSuccessC.tikz}{}{\input{./figures/Diagrams/probSuccessC.tikz}}\ \middle|\ \chi \in \calP_A \right\}.
\end{equation}
In the rest of the paper, we study this quantity.

\section{Designing secure GPT money schemes (and when it is not possible)}

We now study the optimization problem~\eqref{CP} with the hopes of finding bank strategies $\mathcal{S}_A$ which make $\alpha_A$ as small as possible.
We start with a very weak condition which we call Weak No-Counterfeiting which roughly states that a counterfeiter cannot cheat perfectly.
This is of course not useful for cryptographic purposes.
Therefore, we later discuss how to start with this condition and strengthen it by driving the maximum success probability of a counterfeiter down to a negligible amount.
We therefore prove that all GPTs fall into one of two categories: there is either perfect counterfeiting (this is the case for classical theory) or practical security (this is the case for quantum theory). Interestingly, there is no middle ground.

\subsection{Perfect counterfeiting}

In order to understand what we need to achieve security we first consider the converse problem: what does it mean for perfect counterfeiting to be possible?

This seems closely related to the \emph{clonability} or \emph{broadcastability} of the ensemble $\varepsilon_A = \{ (p_1, s_1), \ldots, (p_n, s_n) \}$. Specifically, if one could perfectly clone the states then this would immediately provide a perfect counterfeiting strategy. That is, given a channel $\Delta$ such that, for all $s_i$:
\begin{equation}%
\InputIfFileExists{Diagrams/perfectCloner.tikz}{}{\input{./figures/Diagrams/perfectCloner.tikz}}\ =\ %
\InputIfFileExists{Diagrams/perfectCloned.tikz}{}{\input{./figures/Diagrams/perfectCloned.tikz}}\end{equation}
then if the counterfeiter used this channel we would find that
\begin{equation}
\InputIfFileExists{Diagrams/clonerAsCounterfeiter.tikz}{}{\input{./figures/Diagrams/clonerAsCounterfeiter.tikz}}\ =\ %
\InputIfFileExists{Diagrams/clonerAsCounterfeiter2.tikz}{}{\input{./figures/Diagrams/clonerAsCounterfeiter2.tikz}}\ =\ %
\InputIfFileExists{Diagrams/clonerAsCounterfeiter3.tikz}{}{\input{./figures/Diagrams/clonerAsCounterfeiter3.tikz}}
\ =\ 1 \enperiod
\end{equation}
This implies $\alpha_A = 1$ and so counterfeiting can be achieved perfectly.
This is exactly what happens in the case of classical theory. Moreover, it was shown in \cite{barnum2007generalized} that the possibility of cloning arbitrary ensembles singles out classical theory from a wide class of GPTs\footnote{Namely, those satisfying the No-Restriction Hypothesis \cite{chiribella2010probabilistic} and Tomographic Locality \cite{hardy2001quantum}.}.

However, there are other ways in which one could find that perfect counterfeiting is possible. One way is if the bank strategy is too restricted. For example, if we take
\begin{equation}
\begin{tikzpicture}
	\begin{pgfonlayer}{nodelayer}
		\node [style=none] (0) at (0, 0) {};
		\node [style=none] (1) at (0, -0.5) {};
		\node [style=copoint] (2) at (0, 0.25) {$e_i$};
		\node [style=right label] (3) at (0, -0.5) {$A$};
	\end{pgfonlayer}
	\begin{pgfonlayer}{edgelayer}
		\draw (1.center) to (0.center);
	\end{pgfonlayer}
\end{tikzpicture}
}\ =\ %
\begin{tikzpicture}
	\begin{pgfonlayer}{nodelayer}
		\node [style=none] (0) at (0, -0.5) {};
		\node [style=none] (1) at (0, -0) {};
		\node [style=upground] (2) at (0, 0.25) {};
		\node [style={right label}] (3) at (0, -0.5) {$A$};
	\end{pgfonlayer}
	\begin{pgfonlayer}{edgelayer}
		\draw (1.center) to (0.center);
	\end{pgfonlayer}
\end{tikzpicture}
} \encomma
\end{equation}
for all $i$, then the bank is learning nothing about the returned banknotes. Obviously, the counterfeiter would always pass the verification irrespective of what counterfeiting procedure they used.

Therefore, the task of finding a secure money scheme for the bank does not only require a non-cloneable ensemble $\varepsilon_A = \{ (p_1, s_1), \ldots, (p_n, s_n) \}$, but also a decent choice of each verification effect $e_i$. To this end, we now consider different security notions for the entire bank strategy $\mathcal{S}_A = \{ (p_1, s_1, e_1), \ldots, (p_n, s_n, e_n) \}$.

\subsection{Weak security}

We start with a definition.

\medskip
\begin{defn}[Weak No-Counterfeiting ($\WNC$)]
A theory satisfies $\WNC$ if there exists a bank strategy $\mathcal{S}_A = \{ (p_1, s_1, e_1), \ldots, (p_n, s_n, e_n) \}$ such that $\alpha_A < 1$.
\end{defn}
\medskip

This is mathematically the weakest condition that can be used to rule out perfect counterfeiting. However, it is not the most physically motivated assumption. We  therefore address how it can be derived from an assumption regarding the bank strategy.

\medskip
\begin{defn}[Verification Sharpness ($\VS$)]
The bank strategy $$\mathcal{S}_A = \{ (p_1, s_1, e_1), \ldots, (p_n, s_n, e_n) \}$$ satisfies $\VS$ if and only if the effects uniquely pick out the states in the ensemble. I.e., for all $i$, we have:
\begin{equation}
\InputIfFileExists{Diagrams/verifSharp.tikz}{}{\input{./figures/Diagrams/verifSharp.tikz}}\ =\ 1\quad \iff\quad %
\InputIfFileExists{Diagrams/verifSharp2.tikz}{}{\input{./figures/Diagrams/verifSharp2.tikz}} \enperiod
\end{equation}
{where $s$ is an arbitrary subcausal state.}
\end{defn}

Verification Sharpness is defined to rule out the case discussed in the previous subsection where the bank's measurements were too restricted, e.g., each $e_i$ is just the discarding effect.
The effects in the Verification Sharpness condition are idealised in the sense that any deviation from the honest state will be caught by the bank with non-zero probability.
In quantum theory, if the states are pure, then the effect can be chosen to be the rank-$1$ projection onto that state. Thus, in quantum theory, $\VS$ holds when the states in the ensemble are pure.

\begin{defn}[Broadcasting Map]
A broadcasting map, $B$, for a set of states $\{s_i\}$ is any subcausal map satisfying:
\begin{equation}\label{eq:broadcast}
\InputIfFileExists{Diagrams/broadcast1.tikz}{}{\input{./figures/Diagrams/broadcast1.tikz}} \encomma \quad \forall i \enperiod
\end{equation}
\end{defn}

We now show that under the $\VS$ condition, a violation of $\WNC$ is equivalent to being able to broadcast the ensemble.

\begin{lm}\label{lem:WNC}
Suppose a bank strategy $\mathcal{S}_A = \{ (p_1, s_1, e_1), \ldots, (p_n, s_n, e_n) \}$ satisfies $\VS$. Then perfect counterfeiting is equivalent to the broadcastability of the states $\{ s_1, \ldots, s_n \}$.
\end{lm}

{We postpone a proof to Appendix~\ref{appproof} since it follows easily from (independent) analysis later in the paper.}

In \cite{barnum2007generalized} it was shown that broadcastability (or the special case of clonability) of an ensemble implies it must be classical {(at least in their framework which assumes the No-Restriction Hypothesis \cite{chiribella2010probabilistic} and Tomographic Locality \cite{hardy2001quantum})}. Hence, we obtain the following corollary by restricting our consideration to the class of GPTs that they consider.

{\begin{cor}\label{cor:1}
In any non-classical GPT satisfying the No-Restriction Hypothesis \cite{chiribella2010probabilistic} and Tomographic Locality \cite{hardy2001quantum}, $\VS$ implies that the strategy must satisfy $\WNC$.
\end{cor}
}

Clearly, the $\WNC$ condition is not good enough on its own to have secure money as the counterfeiter could still cheat with very large probability.
In the next subsection, we introduce a variant which is more meaningful for cryptographic security.
For this purpose, we prove that bank strategies satisfying $\WNC$ can be modified such that the states \emph{span} the vector space $V^A$.

\begin{defn}
A bank strategy $\mathcal{S}_A = \{ (p_1, s_1, e_1), \ldots, (p_n, s_n, e_n) \}$ is said to be spanning if $\{ s_1, \ldots, s_n \}$ span the vector space $V^A$.
\end{defn}

We have the following lemma.

\begin{lm} \label{lemma}
If $\WNC$ holds, then there is a spanning bank strategy $\calS_A$ such that $\alpha_A < 1$.
\end{lm}

\begin{proof}
Suppose we have a bank strategy $S'_A$ with security parameter $\alpha'_A < 1$ (which exists since $\WNC$ holds.
Then given a basis of causal states
\footnote{Which exists since the vector space is spanned by the set of causal states.}
$\{b_j\}_{j=1}^m$ of $V^A$, we can construct the bank strategy
$S''_A = \{ (1/m, b_j, \detEff) \}$.
The security parameter for the strategy $S''_A$ is denoted $\alpha''_A$ (which clearly equals $1$).
Let $\mathcal{S}_A$ be the bank strategy which uses $\mathcal{S}_A'$ or $\mathcal{S}_A''$ chosen uniformly at random.
As $\mathcal{S}_A''$ is a spanning bank strategy, we see that $\mathcal{S}_A$ is as well.
The proof now follows since
\begin{equation}
\alpha_A \leq \frac{1}{2} \left( \alpha'_A + \alpha_A'' \right) < 1
\end{equation}
as $\alpha_A' < 1$ and $\alpha_A'' = 1$.
\end{proof}

Thus, for the rest of the paper, we may restrict our attention to spanning bank strategies when requiring $\WNC$ to hold.

\subsection{Practical security}

We start by defining a condition which allows for a practical level of security.

\begin{defn}[Strong No-Counterfeiting ($\SNC$)]
A theory satisfies $\SNC$ if for any $\delta > 0$ there exists a bank strategy $\mathcal{S}_A$ such that ${\alpha_A \leq \delta}$.
\end{defn}

Assuming $\SNC$, the bank can therefore choose a value for $\delta$ with which it is comfortable, then proceed to use the appropriate ensemble in the banknote and the corresponding effects in its verification.
Note also that there is no hope of doing better than this, i.e., we cannot take $\delta=0$ as there is always some probability that the counterfeiter could make a lucky guess and prepare a new note in exactly the right state.
In particular, if $p_i = \max \{ p_1, \ldots, p_n \}$, then the counterfeiter can always use the counterfeiting machine
\begin{equation} \label{trivstrat}
\InputIfFileExists{Diagrams/trivialCounterfeiterNEW.tikz}{}{\input{./figures/Diagrams/trivialCounterfeiterNEW.tikz}}
\end{equation}
which succeeds with probability at least $p_i > 0$.
({Even if a counterfeiter did not know the bank's strategy, they could use the strategy \eqref{trivstrat} but instead of $s_i$, use a state in the interior of the cone to show that $\alpha_A > 0$.})

In this subsection we assume that the theory satisfies $\WNC$ (i.e., there is a bank strategy in the theory which satisfies $\WNC$) and ask if this can be extended to a proof of $\SNC$.
Specifically we consider boosting the security by having multiple independent copies of a bank strategy on each banknote.
For example, let $\mathcal{S}_A$ and $\mathcal{S}_B$ be two bank strategies.
We now study the case if the bank were to use \emph{both} $\mathcal{S}_A$ and $\mathcal{S}_B$ by sampling a state from one, then the other independently, and including both sampled states on the banknote. Let $\mathcal{S}_{AB}$ be the new (product) bank strategy and let $\alpha_{AB}$ be the optimal probability that a counterfeiter can successfully cheat $\mathcal{S}_{AB}$.
It seems like it should be true that $\alpha_{AB} = \alpha_A \alpha_B$.
Indeed this is the case for quantum theory~\cite{MVW13}.
However, there are examples, even in classical theory, where similar tasks do not have such a product theorem (one example can be found in the study of nonlocal games~\cite{feige1992two}).
Therefore, such a result is not so forthcoming.

Indeed, the fact that $\alpha_{AB}$ may not equal $\alpha_A \alpha_B$ is a big difference between the setting of GPTs and quantum theory.
We could impose certain physical principles (each of which hold in quantum theory) in order to enforce $\alpha_{AB} = \alpha_A \alpha_B$, but we will show how to circumvent this problem entirely without the need to assume these extra physical principles.

Let us consider a counterfeiting strategy ${\chi_{AB}\in K_{AB}^{AABB}}$ as illustrated diagrammatically, below
\begin{equation}
\InputIfFileExists{Diagrams/chiAB.tikz}{}{\input{./figures/Diagrams/chiAB.tikz}} \enperiod
\end{equation}

Clearly it is possible to achieve\footnote{We did not prove attainment of an optimal solution in this work, but the same argument holds if one takes limits.} success probability $\alpha_A\alpha_B$ by the counterfeiter by simply doing his optimal procedure on $A$ and $B$ independently:
\begin{equation}
\InputIfFileExists{Diagrams/chiABopt.tikz}{}{\input{./figures/Diagrams/chiABopt.tikz}}
\enperiod
\end{equation}

Ideally we would like to show that they cannot do any better than this, but as mentioned above, this may not always be the case.

We now show that although $\alpha_{AB}$ might be greater than $\alpha_A \alpha_B$, it cannot be too much greater.
To this end, we relax $\alpha_A$ to a new quantity $\wtalpha_A > 0$ (defined in the next subsection) which satisfies the following two properties:
\begin{eqnarray}
& & \alpha_A < 1 \implies \wtalpha_A < 1 \text{ for all spanning bank strategies } \mathcal{S}_A; \label{P1} \\
& & \alpha_{A^{\otimes n}} \leq (\wtalpha_{A})^n \text{ for all (not necessarily spanning) bank strategies } \mathcal{S}_A. \label{P2}
\end{eqnarray}
We see that \eqref{P1} together with Lemma~\ref{lemma} says that if $\WNC$ holds for \emph{any} bank strategy, then this new quantity is bounded away from $1$ for some spanning bank strategy.
Combining this with \eqref{P2}, we have that the success probability of cheating the bank decreases exponentially using the $n$-fold repetition of this spanning bank strategy.  This is summarised in the following lemma.

\medskip
\begin{lm} \label{helpful}
If there exists a quantity $\wtalpha > 0$ satisfying Eqs.~\eqref{P1}~and~\eqref{P2}, then $\WNC$ implies $\SNC$.
\end{lm}

In the following subsection we define a relaxation from $\alpha$ to $\wtalpha$ and demonstrate that it satisfies Eqs.~\eqref{P1}~and~\eqref{P2}. Thus the weakest form of security possible implies the promise of practical security.

\subsubsection{A helpful, possibly non-physical, quantity $\wtalpha$}

{Recall the definition of a dual cone (Definition~\ref{dualcone}).}
The dual cones ${K_A^B}^*$ do not have an immediate interpretation within the GPT. However, it is clear that a diagram of the form
\begin{equation}
\InputIfFileExists{Diagrams/test.tikz}{}{\input{./figures/Diagrams/test.tikz}}
\end{equation}
must be within ${K_A^B}^*$ as it necessarily must evaluate to a non-negative real number on any $f\in K_A^B$. In particular, this implies that $K^A \subseteq {K_A^*}$ and $K_A\subseteq {K^A}^*$, that is, the state cone lives inside the dual of the effect cone, and vice versa.

To define $\wtalpha$, we relax the set of physical counterfeiting machines $\calP$ to $\caltP$, where $\caltP$ is a set of ``counterfeiting machines'' which are possibly not physically realizable.
However, whilst a counterfeiter may not be able to physically perform $\tilde{\chi} \in \caltP$, it is nonetheless useful to consider.

Recall we have $\calP$ defined as
\begin{equation}
\calP_A = \left\{\ \ %
\InputIfFileExists{Diagrams/causalX.tikz}{}{\input{./figures/Diagrams/causalX.tikz}} \ \leq_{K_A} \ %
\begin{tikzpicture}
	\begin{pgfonlayer}{nodelayer}
		\node [style=none] (0) at (0, -0.75) {};
		\node [style=none] (1) at (0, 0.25) {};
		\node [style=upground] (2) at (0, 0.5) {};
	\end{pgfonlayer}
	\begin{pgfonlayer}{edgelayer}
		\draw (1.center) to (0.center);
	\end{pgfonlayer}
\end{tikzpicture}
}, \quad  %
\InputIfFileExists{Diagrams/cloner_jamie.tikz}{}{\input{./figures/Diagrams/cloner_jamie.tikz}}\in K_{A}^{AA} \ \ \right\}
\end{equation}
where we have now explicitly denoted that the set defining the ordering is the effect cone for $A$. This is the cone that we will extend for the relaxation. Specifically, we replace $K_A$ with ${K^A}^*$ which, as mentioned above, contains $K_A$. Therefore, we define the relaxation
\begin{equation}
\caltP_A = \left\{\ \ %
\InputIfFileExists{Diagrams/causalX.tikz}{}{\input{./figures/Diagrams/causalX.tikz}} \ \leq_{{K^A}^*} \ %
\begin{tikzpicture}
	\begin{pgfonlayer}{nodelayer}
		\node [style=none] (0) at (0, -0.75) {};
		\node [style=none] (1) at (0, 0.25) {};
		\node [style=upground] (2) at (0, 0.5) {};
	\end{pgfonlayer}
	\begin{pgfonlayer}{edgelayer}
		\draw (1.center) to (0.center);
	\end{pgfonlayer}
\end{tikzpicture}
}, \quad  %
\InputIfFileExists{Diagrams/cloner_jamie.tikz}{}{\input{./figures/Diagrams/cloner_jamie.tikz}}\in K_{A}^{AA} \ \ \right\}
\end{equation}
which may contain non-physical counterfeiting machines. They must still, by assumption, be in the process cone. But, there is no guarantee that they are causal or even subcausal, therefore they could lead to obtaining `probabilities' greater than~$1$.

This is clearly a relaxation as for any GPT, we have  $K_A\subseteq {K^A}^*$.
However, certain GPTs satisfy a property known as the \emph{No-Restriction Hypothesis} \cite{chiribella2010probabilistic}, which states that $K_A = {K^A}^*$.
For such theories we have $\caltP_A=\calP_A$ and so the relaxation is trivial.
In particular this is the case for quantum theory.
Therefore, the proof of our main result greatly simplifies in the case of quantum theory and any other theory satisfying the No-Restriction Hypothesis.

Define the quantity $\wtalpha$ to be
\begin{equation} \label{primal}
\wtalpha_A = \sup \left\{\  %
\InputIfFileExists{Diagrams/probSuccessC_jamie.tikz}{}{\input{./figures/Diagrams/probSuccessC_jamie.tikz}}\ \middle| \ \ \tilde{\chi} \in \caltP_A \ \right\}
\end{equation}
which is defined to be optimizing the same quantity as $\alpha$, i.e., the same $C$ process.
Also, we have $0 < \alpha \leq \wtalpha$ since $\calP \subseteq \caltP$.
Here, we use tildes to denote something that may not be physical to make the distinction clear.

Below is the main technical result of this paper. The proof of this theorem relies crucially on the duality theory of cone programming. The interested reader is referred to the book~\cite{BV}.

\medskip
\begin{thm} \label{dualitythm}
For all bank strategies $\mathcal{S}_A$, we have
\begin{equation} \label{dual}
\wtalpha_A = \min \left\{\ \ %
\begin{tikzpicture}
	\begin{pgfonlayer}{nodelayer}
		\node [style=none] (0) at (0, -0.25) {};
		\node [style=none] (1) at (0, 0.25) {};
		\node [style=upground] (2) at (0, 0.5) {};
		\node [style=point] (3) at (0, -0.5) {$y$};
	\end{pgfonlayer}
	\begin{pgfonlayer}{edgelayer}
		\draw (1.center) to (0.center);
	\end{pgfonlayer}
\end{tikzpicture}
} \ \ \middle|\quad
\InputIfFileExists{Diagrams/bankStrategyYDef.tikz}{}{\input{./figures/Diagrams/bankStrategyYDef.tikz}}\ - \ %
\InputIfFileExists{Diagrams/combC.tikz}{}{\input{./figures/Diagrams/combC.tikz}} \in {K_{A}^{AA}}^*, \ y \in K^A \right\} \enperiod
\end{equation}
Moreover, \eqref{primal} and \eqref{dual} attain an optimal solution ({hence the use of ``min'' instead of ``inf'' above}). For notational simplicity, we will denote
\begin{equation} \label{Ydef}
\InputIfFileExists{Diagrams/bankStrategyYDef.tikz}{}{\input{./figures/Diagrams/bankStrategyYDef.tikz}} \ =: \  %
\InputIfFileExists{Diagrams/combTtilde.tikz}{}{\input{./figures/Diagrams/combTtilde.tikz}} \enperiod
\end{equation}
We call the original formulation \eqref{primal} the \emph{primal}, and the above formulation \eqref{dual} the \emph{dual}.
\end{thm}

{Before showing a proof, we need the lemma below (whose proof can be found in Appendix~\ref{AppA}).}

\begin{lm} \label{usefulLemma}
For $y \in \mathrm{int}(K^A)$, we have \;
\[%
\InputIfFileExists{Diagrams/bankStrategyYDef.tikz}{}{\input{./figures/Diagrams/bankStrategyYDef.tikz}} \in \mathrm{int}({K_A^{AA}}^*).\]
\end{lm}

{We now prove Theorem~\ref{dualitythm}.}

\proof\label{dualityproof}
We begin by transforming the optimization problem~\eqref{primal} into a statement about vectors in the finite-dimensional vector space $V_A^{AA}$.
We make the vector-diagram associations
\begin{equation} \label{associations}
\begin{tikzpicture}
	\begin{pgfonlayer}{nodelayer}
		\node [style=none] (0) at (0, -0) {$\tilde{\chi}$};
		\node [style=none] (1) at (-0.75, 0.5) {};
		\node [style=none] (2) at (0.75, 0.5) {};
		\node [style=none] (3) at (0.25, -0.5) {};
		\node [style=none] (4) at (-0.25, -0.5) {};
		\node [style=none] (5) at (-0.25, 0.5) {};
		\node [style=none] (6) at (-0.25, 1) {};
		\node [style=none] (7) at (0.25, 0.5) {};
		\node [style=none] (8) at (0.25, 1) {};
		\node [style=none] (9) at (0, -0.5) {};
		\node [style=none] (10) at (0, -1) {};
	\end{pgfonlayer}
	\begin{pgfonlayer}{edgelayer}
		\draw (1.center) to (2.center);
		\draw (2.center) to (3.center);
		\draw (3.center) to (4.center);
		\draw (4.center) to (1.center);
		\draw (6.center) to (5.center);
		\draw (8.center) to (7.center);
		\draw (9.center) to (10.center);
	\end{pgfonlayer}
\end{tikzpicture}
 \; \longleftrightarrow \; X,
\quad
\InputIfFileExists{Diagrams/combC.tikz}{}{\input{./figures/Diagrams/combC.tikz}} \longleftrightarrow C,
\quad
\InputIfFileExists{Diagrams/doubleDiscard.tikz}{}{\input{./figures/Diagrams/doubleDiscard.tikz}} \longleftrightarrow \phi
\quad
\text{ and }
\quad
\begin{tikzpicture}
	\begin{pgfonlayer}{nodelayer}
		\node [style=upground] (0) at (0, 0.25) {};
		\node [style=none] (1) at (0, -0) {};
		\node [style=none] (2) at (0, -0.75) {};
	\end{pgfonlayer}
	\begin{pgfonlayer}{edgelayer}
		\draw (1.center) to (2.center);
	\end{pgfonlayer}
\end{tikzpicture}} \longleftrightarrow b
\end{equation}
such that
\begin{equation}
\inner{C}{X} =
\begin{tikzpicture}
	\begin{pgfonlayer}{nodelayer}
		\node [style=none] (0) at (0, 0.5) {};
		\node [style=none] (1) at (0, 1) {};
		\node [style=none] (2) at (-0.25, -1) {};
		\node [style=none] (3) at (-0.5, 0.5) {};
		\node [style=none] (4) at (-0.5, 1) {};
		\node [style=none] (5) at (-0.25, -0.5) {};
		\node [style=none] (6) at (0.5, 1) {};
		\node [style=none] (7) at (-1.25, 1) {};
		\node [style=none] (8) at (-1.25, -1) {};
		\node [style=none] (9) at (0.5, -1) {};
		\node [style=none] (10) at (0.5, -1.5) {};
		\node [style=none] (11) at (-2.25, -1.5) {};
		\node [style=none] (12) at (-2.25, 1.5) {};
		\node [style=none] (13) at (0.5, 1.5) {};
		\node [style=none] (14) at (-1.75, -0) {$C$};
		\node [style=none] (15) at (-1, 0.5) {};
		\node [style=none] (16) at (0.5, 0.5) {};
		\node [style=none] (17) at (0, -0.5) {};
		\node [style=none] (18) at (-0.5, -0.5) {};
		\node [style=none] (19) at (-0.25, -0) {$\tilde{\chi}$};
		\node [style=none] (20) at (-0.5, 2) {};
		\node [style=none] (21) at (-0.5, -2) {};
	\end{pgfonlayer}
	\begin{pgfonlayer}{edgelayer}
		\draw (4.center) to (3.center);
		\draw (1.center) to (0.center);
		\draw (5.center) to (2.center);
		\draw (12.center) to (13.center);
		\draw (13.center) to (6.center);
		\draw (6.center) to (7.center);
		\draw (7.center) to (8.center);
		\draw (8.center) to (9.center);
		\draw (9.center) to (10.center);
		\draw (10.center) to (11.center);
		\draw (11.center) to (12.center);
		\draw (15.center) to (16.center);
		\draw (16.center) to (17.center);
		\draw (17.center) to (18.center);
		\draw (18.center) to (15.center);
	\end{pgfonlayer}
\end{tikzpicture}
\quad
\text{,}
\quad
b-\phi(X)\ =\ %
\begin{tikzpicture}
	\begin{pgfonlayer}{nodelayer}
		\node [style=upground] (0) at (0, 0.25) {};
		\node [style=none] (1) at (0, -0) {};
		\node [style=none] (2) at (0, -0.75) {};
	\end{pgfonlayer}
	\begin{pgfonlayer}{edgelayer}
		\draw (1.center) to (2.center);
	\end{pgfonlayer}
\end{tikzpicture}} - %
\InputIfFileExists{Diagrams/vector.tikz}{}{\input{./figures/Diagrams/vector.tikz}}\quad\text{and}\quad\phi^*(y)\ =\ %
\InputIfFileExists{Diagrams/phiStarY.tikz}{}{\input{./figures/Diagrams/phiStarY.tikz}} \enperiod
\end{equation}
Note these exist by the Riesz-Fr\'echet Representation Theorem.
Thus, \eqref{primal} can be written in vector form as in~\eqref{primal2} and its dual \eqref{dual} as in~\eqref{dual2} when $K_1 = K^A$ and $K_2 = K_A^{AA}$.

We now show that strong duality holds.
{As discussed in the paragraph following \eqref{eq:linear}, there are subcausal processes in the interior of the cone. That is,} there exists $\chi \in \mathrm{int}(K_A^{AA})$ satisfying
\begin{equation}
\InputIfFileExists{Diagrams/causalX.tikz}{}{\input{./figures/Diagrams/causalX.tikz}} \leq %
\begin{tikzpicture}
	\begin{pgfonlayer}{nodelayer}
		\node [style=none] (0) at (0, -0.75) {};
		\node [style=none] (1) at (0, 0.25) {};
		\node [style=upground] (2) at (0, 0.5) {};
	\end{pgfonlayer}
	\begin{pgfonlayer}{edgelayer}
		\draw (1.center) to (0.center);
	\end{pgfonlayer}
\end{tikzpicture}
} \enperiod
\end{equation}
{Since $%
\begin{tikzpicture}
	\begin{pgfonlayer}{nodelayer}
		\node [style=upground] (0) at (0, 0.25) {};
		\node [style=none] (1) at (0, -0) {};
		\node [style=none] (2) at (0, -0.75) {};
	\end{pgfonlayer}
	\begin{pgfonlayer}{edgelayer}
		\draw (1.center) to (2.center);
	\end{pgfonlayer}
\end{tikzpicture}} \in \mathrm{int}(K_A)$, we know that \[%
\begin{tikzpicture}
	\begin{pgfonlayer}{nodelayer}
		\node [style=upground] (0) at (0, 0.25) {};
		\node [style=none] (1) at (0, -0) {};
		\node [style=none] (2) at (0, -0.75) {};
	\end{pgfonlayer}
	\begin{pgfonlayer}{edgelayer}
		\draw (1.center) to (2.center);
	\end{pgfonlayer}
\end{tikzpicture}} - \lambda\ %
\begin{tikzpicture}
	\begin{pgfonlayer}{nodelayer}
		\node [style=copoint] (0) at (0, 0.25) {$\nu$};
		\node [style=none] (1) at (0, -0) {};
		\node [style=none] (2) at (0, -0.75) {};
	\end{pgfonlayer}
	\begin{pgfonlayer}{edgelayer}
		\draw (1.center) to (2.center);
	\end{pgfonlayer}
\end{tikzpicture}
} \in \mathrm{int}(K_A)\] for any $\nu \in V_A$ and sufficiently small $\lambda > 0$.
Thus, there exists a $\lambda > 0$ such that
$\lambda \chi \in \mathrm{int}(K_A^{AA})$ and
\begin{equation}
\begin{tikzpicture}
	\begin{pgfonlayer}{nodelayer}
		\node [style=upground] (0) at (0, 0.25) {};
		\node [style=none] (1) at (0, -0) {};
		\node [style=none] (2) at (0, -0.75) {};
	\end{pgfonlayer}
	\begin{pgfonlayer}{edgelayer}
		\draw (1.center) to (2.center);
	\end{pgfonlayer}
\end{tikzpicture}} - \lambda %
\InputIfFileExists{Diagrams/causalX.tikz}{}{\input{./figures/Diagrams/causalX.tikz}} \in \mathrm{int}(K_A) \subseteq \mathrm{int}({K^A}^*).
\end{equation}
Here, we can think of $\lambda$ as a scaling factor to pull the processes into the interior of the necessary cones.
}

{
On the other hand, By Lemma~\ref{usefulLemma}, we can take $y \in \mathrm{int}(K^A)$ to get that \[%
\InputIfFileExists{Diagrams/bankStrategyYDef.tikz}{}{\input{./figures/Diagrams/bankStrategyYDef.tikz}} \in \mathrm{int}({K_A^{AA}}^*).\]
We can repeat the argument above to scale $C$ down (say, by $\mu > 0$) to get that
\begin{equation}
\InputIfFileExists{Diagrams/phiStarY.tikz}{}{\input{./figures/Diagrams/phiStarY.tikz}} - \mu \, %
\InputIfFileExists{Diagrams/combC.tikz}{}{\input{./figures/Diagrams/combC.tikz}}\in\mathrm{int}({K_A^{AA}}^*).
\end{equation}
Define $\lambda' = 1/\mu$ to get
\begin{equation}
\lambda'\ %
\begin{tikzpicture}
	\begin{pgfonlayer}{nodelayer}
		\node [style=none] (0) at (0, -0) {};
		\node [style=none] (1) at (0, 1) {};
		\node [style=point] (2) at (0, -0.25) {$y$};
	\end{pgfonlayer}
	\begin{pgfonlayer}{edgelayer}
		\draw (1.center) to (0.center);
	\end{pgfonlayer}
\end{tikzpicture}} \in \mathrm{int}(K^A)\quad\text{and}\quad \lambda'\ %
\InputIfFileExists{Diagrams/phiStarY.tikz}{}{\input{./figures/Diagrams/phiStarY.tikz}}-%
\InputIfFileExists{Diagrams/combC.tikz}{}{\input{./figures/Diagrams/combC.tikz}}\in\mathrm{int}({K_A^{AA}}^*).
\end{equation}
}

This proves that strong duality holds {(recall the definition from Section~\ref{convex})} implying that the primal and dual have the same optimal value and both problems attain an optimal solution.
\endproof


{
\begin{remark} \label{rem:alpha}
From the proof above, we see that the value $\alpha$ (recall \eqref{CP}) is attained as well. This is because the above proof easily generalizes to the case where we replace ${K^A}^*$ with $K_A$ (even the interior points are the same).
\end{remark}
}

We now prove $\wtalpha$ satisfies the properties required for Lemma~\ref{helpful}.

\begin{lm}\label{lem4}
For all spanning bank strategies $\calS_A$, we have $\wtalpha_A = 1$ implies that $\alpha_A = 1$, {i.e.} \eqref{P1} is satisfied.
\end{lm}

\begin{proof}
Notice we have
\begin{equation} \label{bound}
\InputIfFileExists{Diagrams/bound.tikz}{}{\input{./figures/Diagrams/bound.tikz}}
\quad \leq \quad
\sum_i p_i %
\InputIfFileExists{Diagrams/bound1.tikz}{}{\input{./figures/Diagrams/bound1.tikz}}
 \quad \leq\quad \sum_i p_i\ =\ 1
\end{equation}
for all $\tilde{\chi} \in \caltP$.
Thus we have $\wtalpha_A \leq 1$.
Let $\tilde{\chi} \in \caltP$ be an optimal solution to \eqref{primal} and suppose $\wtalpha_A = 1$.
We see from \eqref{bound} that
\begin{equation}
\InputIfFileExists{Diagrams/bound2.tikz}{}{\input{./figures/Diagrams/bound2.tikz}}\ =\ 1 \encomma
\end{equation}
for all $i$. As the states $s_i$ are assumed to span $V^A$ this implies {for all states $s$ (and indeed for any vector $s$) that
\[%
\InputIfFileExists{Diagrams/chiCausal1.tikz}{}{\input{./figures/Diagrams/chiCausal1.tikz}} = %
\begin{tikzpicture}
	\begin{pgfonlayer}{nodelayer}
		\node [style=none] (0) at (0, -0.25) {};
		\node [style=none] (1) at (0, 0.25) {};
		\node [style=upground] (2) at (0, 0.5) {};
		\node [style=point] (3) at (0, -0.5) {$s$};
	\end{pgfonlayer}
	\begin{pgfonlayer}{edgelayer}
		\draw (1.center) to (0.center);
	\end{pgfonlayer}
\end{tikzpicture}}\]
and hence, via tomography \eqref{eq:Tomog}, that
\[%
\InputIfFileExists{Diagrams/chiCausal4.tikz}{}{\input{./figures/Diagrams/chiCausal4.tikz}} = %
\begin{tikzpicture}
	\begin{pgfonlayer}{nodelayer}
		\node [style=none] (0) at (0, -0.25) {};
		\node [style=none] (1) at (0, 0.25) {};
		\node [style=upground] (2) at (0, 0.5) {};
	\end{pgfonlayer}
	\begin{pgfonlayer}{edgelayer}
		\draw (1.center) to (0.center);
	\end{pgfonlayer}
\end{tikzpicture}} \enperiod \]
In other words} $\tilde{\chi}$ is causal and thus must be in $\calP$ as well.
Since there exists $\chi \in \calP$ such that
\begin{equation}
\InputIfFileExists{Diagrams/bound3.tikz}{}{\input{./figures/Diagrams/bound3.tikz}}\ =\ 1 \encomma
\end{equation}
we have $\alpha_A = 1$ as {desired}.
\end{proof}

We now consider the dual problem.
Consider an optimal solution to the dual problem $y \in K^A$, such that
\begin{equation}
\wtalpha\ =\ %
\begin{tikzpicture}
	\begin{pgfonlayer}{nodelayer}
		\node [style=none] (0) at (0, -0.25) {};
		\node [style=none] (1) at (0, 0.25) {};
		\node [style=upground] (2) at (0, 0.5) {};
		\node [style=point] (3) at (0, -0.5) {$y$};
	\end{pgfonlayer}
	\begin{pgfonlayer}{edgelayer}
		\draw (1.center) to (0.center);
	\end{pgfonlayer}
\end{tikzpicture}
} \enperiod
\end{equation}
Notice that $Y := \dfrac{1}{\wtalpha} \tilde{Y}$, given diagrammatically as
\begin{equation}
\frac{1}{\wtalpha}\ %
\InputIfFileExists{Diagrams/bankStrategyYDef.tikz}{}{\input{./figures/Diagrams/bankStrategyYDef.tikz}}
\end{equation}
is a physical, causal, operation as
\begin{equation}
\frac{1}{\wtalpha}\ %
\begin{tikzpicture}
	\begin{pgfonlayer}{nodelayer}
		\node [style=none] (0) at (0, -0) {};
		\node [style=none] (1) at (0, 1) {};
		\node [style=point] (2) at (0, -0.25) {$y$};
	\end{pgfonlayer}
	\begin{pgfonlayer}{edgelayer}
		\draw (1.center) to (0.center);
	\end{pgfonlayer}
\end{tikzpicture}}\ \ =\ \ \frac{1}{%
\begin{tikzpicture}
	\begin{pgfonlayer}{nodelayer}
		\node [style=none] (0) at (0, -0.25) {};
		\node [style=none] (1) at (0, 0.25) {};
		\node [style=upground] (2) at (0, 0.5) {};
		\node [style=point] (3) at (0, -0.5) {$y$};
	\end{pgfonlayer}
	\begin{pgfonlayer}{edgelayer}
		\draw (1.center) to (0.center);
	\end{pgfonlayer}
\end{tikzpicture}
}}\quad %
\begin{tikzpicture}
	\begin{pgfonlayer}{nodelayer}
		\node [style=none] (0) at (0, -0) {};
		\node [style=none] (1) at (0, 1) {};
		\node [style=point] (2) at (0, -0.25) {$y$};
	\end{pgfonlayer}
	\begin{pgfonlayer}{edgelayer}
		\draw (1.center) to (0.center);
	\end{pgfonlayer}
\end{tikzpicture}}
\end{equation}
is a physical, i.e. normalised, state.

\begin{remark}
The entire point of relaxing $\alpha$ to $\wtalpha$ was so that $y$ is in $K^A$, thus implying that $Y$ is a \emph{physical}, or causal, operation. We soon show that this is the keystone to our proof of $\SNC$ holding.
\end{remark}

Note that the constraint
\begin{equation}
{\wtalpha} \; %
\InputIfFileExists{Diagrams/combY.tikz}{}{\input{./figures/Diagrams/combY.tikz}} - \ %
\InputIfFileExists{Diagrams/combC.tikz}{}{\input{./figures/Diagrams/combC.tikz}}\ \ =\ \ %
\InputIfFileExists{Diagrams/constraint.tikz}{}{\input{./figures/Diagrams/constraint.tikz}}
 - \ %
\InputIfFileExists{Diagrams/combC.tikz}{}{\input{./figures/Diagrams/combC.tikz}}\ \  \in\ \ {K_{A}^{AA}}^*
\end{equation}
implies that
\begin{equation}\label{eq:YBoundC}
{\wtalpha} \; \begin{tikzpicture}
    \begin{pgfonlayer}{nodelayer}
		\node [style=none] (0) at (-0.25, 0.5) {};
		\node [style=none] (1) at (-0.25, 1) {};
		\node [style=none] (2) at (-0.5, -1) {};
		\node [style=none] (3) at (-0.75, 0.5) {};
		\node [style=none] (4) at (-0.75, 1) {};
		\node [style=none] (5) at (-0.5, -0.5) {};
		\node [style=none] (6) at (0.25, 1) {};
		\node [style=none] (7) at (-1.5, 1) {};
		\node [style=none] (8) at (-1.5, -1) {};
		\node [style=none] (9) at (0.25, -1) {};
		\node [style=none] (10) at (0.25, -1.5) {};
		\node [style=none] (11) at (-2.5, -1.5) {};
		\node [style=none] (12) at (-2.5, 1.5) {};
		\node [style=none] (13) at (0.25, 1.5) {};
		\node [style=none] (14) at (-2, -0) {$Y$};
		\node [style=none] (15) at (-1.25, 0.5) {};
		\node [style=none] (16) at (-0.75, -0.5) {};
		\node [style=none] (17) at (-0.25, -0.5) {};
		\node [style=none] (18) at (0.25, 0.5) {};
		\node [style=none] (19) at (-0.5, -0) {$\xi$};
	\end{pgfonlayer}
	\begin{pgfonlayer}{edgelayer}
		\draw (4.center) to (3.center);
		\draw (1.center) to (0.center);
		\draw (5.center) to (2.center);
		\draw (12.center) to (13.center);
		\draw (13.center) to (6.center);
		\draw (6.center) to (7.center);
		\draw (7.center) to (8.center);
		\draw (8.center) to (9.center);
		\draw (9.center) to (10.center);
		\draw (10.center) to (11.center);
		\draw (11.center) to (12.center);
		\draw (15.center) to (18.center);
		\draw (18.center) to (17.center);
		\draw (17.center) to (16.center);
		\draw (16.center) to (15.center);
	\end{pgfonlayer}
\end{tikzpicture} \ \ \geq\ \ \begin{tikzpicture}
	\begin{pgfonlayer}{nodelayer}
		\node [style=none] (0) at (-0.25, 0.5) {};
		\node [style=none] (1) at (-0.25, 1) {};
		\node [style=none] (2) at (-0.5, -1) {};
		\node [style=none] (3) at (-0.75, 0.5) {};
		\node [style=none] (4) at (-0.75, 1) {};
		\node [style=none] (5) at (-0.5, -0.5) {};
		\node [style=none] (6) at (0.25, 1) {};
		\node [style=none] (7) at (-1.5, 1) {};
		\node [style=none] (8) at (-1.5, -1) {};
		\node [style=none] (9) at (0.25, -1) {};
		\node [style=none] (10) at (0.25, -1.5) {};
		\node [style=none] (11) at (-2.5, -1.5) {};
		\node [style=none] (12) at (-2.5, 1.5) {};
		\node [style=none] (13) at (0.25, 1.5) {};
		\node [style=none] (14) at (-2, -0) {$C$};
		\node [style=none] (15) at (-1.25, 0.5) {};
		\node [style=none] (16) at (-0.75, -0.5) {};
		\node [style=none] (17) at (-0.25, -0.5) {};
		\node [style=none] (18) at (0.25, 0.5) {};
		\node [style=none] (19) at (-0.5, -0) {$\xi$};
	\end{pgfonlayer}
	\begin{pgfonlayer}{edgelayer}
		\draw (4.center) to (3.center);
		\draw (1.center) to (0.center);
		\draw (5.center) to (2.center);
		\draw (12.center) to (13.center);
		\draw (13.center) to (6.center);
		\draw (6.center) to (7.center);
		\draw (7.center) to (8.center);
		\draw (8.center) to (9.center);
		\draw (9.center) to (10.center);
		\draw (10.center) to (11.center);
		\draw (11.center) to (12.center);
		\draw (15.center) to (18.center);
		\draw (18.center) to (17.center);
		\draw (17.center) to (16.center);
		\draw (16.center) to (15.center);
	\end{pgfonlayer}
\end{tikzpicture}
 \encomma \qquad \forall \xi \in {K_{A}^{AA}} \enperiod
\end{equation}

Now let us consider appending a new system $B$ and a new counterfeiting machine which maps $AB$ to the space $AABB$. Diagrammatically, it would look like this
\begin{equation}
\begin{tikzpicture}
	\begin{pgfonlayer}{nodelayer}
		\node [style=none] (0) at (2.25, 0.5) {};
		\node [style=none] (1) at (2.25, 1.25) {};
		\node [style=none] (2) at (1.5, -1.25) {};
		\node [style=none] (3) at (-0.25, 0.5) {};
		\node [style=none] (4) at (-0.25, 1.25) {};
		\node [style=none] (5) at (1.5, -0.5) {};
		\node [style=none] (6) at (-0.75, 0.5) {};
		\node [style=none] (7) at (0, -0.5) {};
		\node [style=none] (8) at (2, -0.5) {};
		\node [style=none] (9) at (2.75, 0.5) {};
		\node [style=none] (10) at (1, -0) {$\chi$};
		\node [style=none] (11) at (0.5, 1.25) {};
		\node [style=none] (12) at (0.5, 0.5) {};
		\node [style=none] (13) at (1.5, 1.25) {};
		\node [style=none] (14) at (1.5, 0.5) {};
		\node [style=none] (15) at (0.5, -0.5) {};
		\node [style=none] (16) at (0.5, -1.25) {};
		\node [style={right label}] (17) at (-0.25, 0.75) {$A$};
		\node [style={right label}] (18) at (0.5, 0.75) {$A$};
		\node [style={right label}] (19) at (0.5, -1) {$A$};
		\node [style={right label}] (20) at (1.5, 0.75) {$B$};
		\node [style={right label}] (21) at (2.25, 0.75) {$B$};
		\node [style={right label}] (22) at (1.5, -1) {$B$};
	\end{pgfonlayer}
	\begin{pgfonlayer}{edgelayer}
		\draw (4.center) to (3.center);
		\draw (1.center) to (0.center);
		\draw (5.center) to (2.center);
		\draw (6.center) to (9.center);
		\draw (9.center) to (8.center);
		\draw (8.center) to (7.center);
		\draw (7.center) to (6.center);
		\draw (11.center) to (12.center);
		\draw (13.center) to (14.center);
		\draw (15.center) to (16.center);
	\end{pgfonlayer}
\end{tikzpicture}.
\end{equation}
Note that for any $\chi_{AB} \in \calP_{AB}$, we have that
\begin{equation} \label{Dpic}
\begin{tikzpicture}
	\begin{pgfonlayer}{nodelayer}
		\node [style=none] (0) at (2.25, 0.5) {};
		\node [style=none] (1) at (2.25, 1.25) {};
		\node [style=none] (2) at (1.5, -1.25) {};
		\node [style=none] (3) at (-0.25, 0.5) {};
		\node [style=none] (4) at (-0.25, 2) {};
		\node [style=none] (5) at (1.5, -0.5) {};
		\node [style=none] (6) at (-0.75, 0.5) {};
		\node [style=none] (7) at (0, -0.5) {};
		\node [style=none] (8) at (2, -0.5) {};
		\node [style=none] (9) at (2.75, 0.5) {};
		\node [style=none] (10) at (1, -0) {$\chi_{AB}$};
		\node [style=none] (11) at (0.5, 2) {};
		\node [style=none] (12) at (0.5, 0.5) {};
		\node [style=none] (13) at (1.5, 1.25) {};
		\node [style=none] (14) at (1.5, 0.5) {};
		\node [style=none] (15) at (0.5, -0.5) {};
		\node [style=none] (16) at (0.5, -2) {};
		\node [style={right label}] (17) at (-0.25, 0.75) {$A$};
		\node [style={right label}] (18) at (0.5, 0.75) {$A$};
		\node [style={right label}] (19) at (0.5, -1) {$A$};
		\node [style={right label}] (20) at (1.5, 0.75) {$B$};
		\node [style={right label}] (21) at (2.25, 0.75) {$B$};
		\node [style={right label}] (22) at (1.5, -1) {$B$};
		\node [style=none] (23) at (1, 1.25) {};
		\node [style=none] (24) at (1, 1.75) {};
		\node [style=none] (25) at (3.25, 1.25) {};
		\node [style=none] (26) at (3.25, -1.25) {};
		\node [style=none] (27) at (1, -1.25) {};
		\node [style=none] (28) at (1, -1.75) {};
		\node [style=none] (29) at (4.25, -1.75) {};
		\node [style=none] (30) at (4.25, 1.75) {};
		\node [style=none] (31) at (3.75, -0) {$D$};
	\end{pgfonlayer}
	\begin{pgfonlayer}{edgelayer}
		\draw (4.center) to (3.center);
		\draw (1.center) to (0.center);
		\draw (5.center) to (2.center);
		\draw (6.center) to (9.center);
		\draw (9.center) to (8.center);
		\draw (8.center) to (7.center);
		\draw (7.center) to (6.center);
		\draw (11.center) to (12.center);
		\draw (13.center) to (14.center);
		\draw (15.center) to (16.center);
		\draw (24.center) to (23.center);
		\draw (23.center) to (25.center);
		\draw (25.center) to (26.center);
		\draw (26.center) to (27.center);
		\draw (27.center) to (28.center);
		\draw (28.center) to (29.center);
		\draw (29.center) to (30.center);
		\draw (30.center) to (24.center);
	\end{pgfonlayer}
\end{tikzpicture}
\end{equation}
is in $\calP_A$ for any physical map $D$. This is because it is physically possible for a counterfeiter to do this, and thus must be captured by the set of physical processes $\calP$. We use two different choices for $D$, the bank's strategy $C$, and the alternative strategy $Y$, above.

With these ideas in hand, we can prove the following lemma.

\medskip
\begin{lm} \label{product}
$\alpha_{A_1\cdots A_m} \leq \Pi_{i=1}^m \ \wtalpha_{A_i}$, for all bank strategies $\mathcal{S}_{A_1}, \ldots, \mathcal{S}_{A_m}$. In particular, \eqref{P2} is satisfied.
\end{lm}

\begin{proof}
We prove the $m=2$ case for clarity, but the general case follows by repeating the same argument. We consider a pair of bank strategies, $\mathcal{S}_A$ and $\mathcal{S}_B$ and show that $\alpha_{AB} \leq \wtalpha_A \wtalpha_B$.
Consider
\begin{equation}
\alpha_{AB} \ =\ %
\InputIfFileExists{Diagrams/productForgingProb.tikz}{}{\input{./figures/Diagrams/productForgingProb.tikz}}
\end{equation}
{where $\chi_{AB} \in \calP$ is an optimal solution to \eqref{CP} (which exists by Remark~\ref{rem:alpha})}.
By a trivial diagrammatic rewrite, this is equivalent to
\begin{equation}
\alpha_{AB} \ =\ %
\InputIfFileExists{Diagrams/optimalProductForgingProb2.tikz}{}{\input{./figures/Diagrams/optimalProductForgingProb2.tikz}} \enperiod
\end{equation}
As discussed previously, we have the dotted part in the diagram above is in $K_A^{AA}$. Therefore, from \eqref{eq:YBoundC}, we have
\begin{equation}
\alpha_{AB} \ \leq\  {\wtalpha_{A}} \; %
\InputIfFileExists{Diagrams/optimalProductForgingProb25.tikz}{}{\input{./figures/Diagrams/optimalProductForgingProb25.tikz}} \enperiod
\end{equation}
By another diagrammatic rewrite, we have
\begin{equation}
\alpha_{AB} \ \leq\  {\wtalpha_{A}} \; \begin{tikzpicture}
	\begin{pgfonlayer}{nodelayer}
		\node [style=none] (0) at (2.5, 0.5) {};
		\node [style=none] (1) at (2.5, 2) {};
		\node [style=none] (2) at (2, -2) {};
		\node [style=none] (3) at (-0.5, 0.5) {};
		\node [style=none] (4) at (-0.5, 1) {};
		\node [style=none] (5) at (2, -0.5) {};
		\node [style=none] (6) at (-1, 0.5) {};
		\node [style=none] (7) at (-0.5, -0.5) {};
		\node [style=none] (8) at (2.5, -0.5) {};
		\node [style=none] (9) at (3, 0.5) {};
		\node [style=none] (10) at (1, -0) {${\chi}_{AB}$};
		\node [style=none] (11) at (0.25, 1) {};
		\node [style=none] (12) at (0.25, 0.5) {};
		\node [style=none] (13) at (1.75, 2) {};
		\node [style=none] (14) at (1.75, 0.5) {};
		\node [style=none] (15) at (0, -0.5) {};
		\node [style=none] (16) at (0, -1) {};
		\node [style=none] (17) at (0.75, 1) {};
		\node [style=none] (18) at (1.25, 2) {};
		\node [style=none] (19) at (1.25, 2.5) {};
		\node [style=none] (20) at (3.5, 2) {};
		\node [style=none] (21) at (3.5, -2) {};
		\node [style=none] (22) at (1.5, -2) {};
		\node [style=none] (23) at (1.5, -2.5) {};
		\node [style=none] (24) at (4.5, -2.5) {};
		\node [style=none] (25) at (4.5, 2.5) {};
		\node [style=none] (26) at (0.75, 1.5) {};
		\node [style=none] (27) at (-1.5, 1) {};
		\node [style=none] (28) at (-2.5, 1.5) {};
		\node [style=none] (29) at (-1.5, -1) {};
		\node [style=none] (30) at (-2.5, -1.5) {};
		\node [style=none] (31) at (0.5, -1) {};
		\node [style=none] (32) at (0.5, -1.5) {};
		\node [style=none] (33) at (-2, -0) {${Y}_A$};
		\node [style=none] (34) at (4, -0) {$C_B$};
		\node [style=none] (35) at (-2.75, 1.75) {};
		\node [style=none] (36) at (-2.75, -1.75) {};
		\node [style=none] (37) at (3.25, -1.75) {};
		\node [style=none] (38) at (3.25, 1.75) {};
	\end{pgfonlayer}
	\begin{pgfonlayer}{edgelayer}
		\draw (4.center) to (3.center);
		\draw (1.center) to (0.center);
		\draw (5.center) to (2.center);
		\draw (6.center) to (9.center);
		\draw (9.center) to (8.center);
		\draw (8.center) to (7.center);
		\draw (7.center) to (6.center);
		\draw (11.center) to (12.center);
		\draw (13.center) to (14.center);
		\draw (15.center) to (16.center);
		\draw (28.center) to (26.center);
		\draw (26.center) to (17.center);
		\draw (17.center) to (27.center);
		\draw (27.center) to (29.center);
		\draw (28.center) to (30.center);
		\draw (19.center) to (18.center);
		\draw (18.center) to (20.center);
		\draw (20.center) to (21.center);
		\draw (21.center) to (22.center);
		\draw (22.center) to (23.center);
		\draw (23.center) to (24.center);
		\draw (24.center) to (25.center);
		\draw (25.center) to (19.center);
		\draw (29.center) to (31.center);
		\draw (31.center) to (32.center);
		\draw (32.center) to (30.center);
		\draw [style={thick gray dashed edge}] (35.center) to (38.center);
		\draw [style={thick gray dashed edge}] (38.center) to (37.center);
		\draw [style={thick gray dashed edge}] (35.center) to (36.center);
		\draw [style={thick gray dashed edge}] (36.center) to (37.center);
	\end{pgfonlayer}
\end{tikzpicture} \enperiod
\end{equation}
By repeating the same argument noting the dotted part above is physical, we have
\begin{equation}
\alpha_{AB} \ \leq\ {\wtalpha_{A} \wtalpha_{B}} \;  %
\InputIfFileExists{Diagrams/productForgingProbY.tikz}{}{\input{./figures/Diagrams/productForgingProbY.tikz}} \enperiod
\end{equation}
From the definition of $Y$, it is clear that
\begin{equation}
\InputIfFileExists{Diagrams/productForgingProbY.tikz}{}{\input{./figures/Diagrams/productForgingProbY.tikz}} \; = \ \frac{1}{\wtalpha_A\wtalpha_B}\ \begin{tikzpicture}
	\begin{pgfonlayer}{nodelayer}
		\node [style=none] (0) at (2.5, 0.5) {};
		\node [style=none] (1) at (2.5, 1) {};
		\node [style=none] (2) at (2, -1) {};
		\node [style=none] (3) at (-0.5, 0.5) {};
		\node [style=none] (4) at (-0.5, 1) {};
		\node [style=none] (5) at (2, -0.5) {};
		\node [style=none] (6) at (-1, 0.5) {};
		\node [style=none] (7) at (-0.5, -0.5) {};
		\node [style=none] (8) at (2.5, -0.5) {};
		\node [style=none] (9) at (3, 0.5) {};
		\node [style=none] (10) at (1, -0) {${\chi}_{AB}$};
		\node [style=none] (11) at (0.25, 1) {};
		\node [style=none] (12) at (0.25, 0.5) {};
		\node [style=none] (13) at (1.75, 1) {};
		\node [style=none] (14) at (1.75, 0.5) {};
		\node [style=none] (15) at (0, -0.5) {};
		\node [style=none] (16) at (0, -1) {};
		\node [style=point] (17) at (0, -1.25) {$y^A$};
		\node [style=point] (18) at (2, -1.25) {$y^B$};
		\node [style=upground] (19) at (-0.5, 1.25) {};
		\node [style=upground] (20) at (0.25, 1.25) {};
		\node [style=upground] (21) at (1.75, 1.25) {};
		\node [style=upground] (22) at (2.5, 1.25) {};
	\end{pgfonlayer}
	\begin{pgfonlayer}{edgelayer}
		\draw (4.center) to (3.center);
		\draw (1.center) to (0.center);
		\draw (5.center) to (2.center);
		\draw (6.center) to (9.center);
		\draw (9.center) to (8.center);
		\draw (8.center) to (7.center);
		\draw (7.center) to (6.center);
		\draw (11.center) to (12.center);
		\draw (13.center) to (14.center);
		\draw (15.center) to (16.center);
	\end{pgfonlayer}
\end{tikzpicture}\ \leq \ \frac{1}{\wtalpha_A\wtalpha_B}%
\begin{tikzpicture}
	\begin{pgfonlayer}{nodelayer}
		\node [style=none] (0) at (0, -0.25) {};
		\node [style=none] (1) at (0, 0.25) {};
		\node [style=upground] (2) at (0, 0.5) {};
		\node [style=point] (3) at (0, -0.5) {$y^A$};
	\end{pgfonlayer}
	\begin{pgfonlayer}{edgelayer}
		\draw (1.center) to (0.center);
	\end{pgfonlayer}
\end{tikzpicture}
}%
\begin{tikzpicture}
	\begin{pgfonlayer}{nodelayer}
		\node [style=none] (0) at (0, -0.25) {};
		\node [style=none] (1) at (0, 0.25) {};
		\node [style=upground] (2) at (0, 0.5) {};
		\node [style=point] (3) at (0, -0.5) {$y^B$};
	\end{pgfonlayer}
	\begin{pgfonlayer}{edgelayer}
		\draw (1.center) to (0.center);
	\end{pgfonlayer}
\end{tikzpicture}
} = \; 1 \encomma
\end{equation}
since $\chi_{AB}$ is subcausal.
This finishes the proof.
\end{proof}

The intuition behind this proof is that a counterfeiter could implement the $Y_B$ or $C_B$ strategy him/herself, and this should not allow him/her to pass the verification on $A$ with higher probability.
Also, a counterfeiter trying to cheat the $Y$ strategies is pointless since the $Y$ strategy simply discards the output of any counterfeiting machine.

{One thing to note in the proof of Lemma~\ref{product} is that we used $\chi \in \calP$ and not a $\tilde{\chi} \in \caltP$. This is because the combination of $\chi$ and a physical map $D$ must be in the cone $K_A^{AA}$ (see \eqref{Dpic}). One can show that this is the case for any $\tilde{\chi} \in \caltP$, as $\tilde{\chi}$ can be rescaled to belong to $\calP$, hence we can repeat the proof to show that $\wtalpha$ is multiplicative over the theory. This is a stronger claim than the one in the lemma but is not necessary for our main result.}

By combining Lemmas~\ref{helpful}, \ref{lem4}, and \ref{product}, we have the following theorem.

\begin{thm}
If $\WNC$ holds, then $\SNC$ holds.
In other words, $\WNC$ is necessary and sufficient to make unforgeable money.
\end{thm}

{By combining this with Corollary~\ref{cor:1} we have the following corollary.
\begin{cor}\label{cor:2}
In any non-classical GPT satisfying the No-Restriction Hypothesis \cite{chiribella2010probabilistic} and Tomographic Locality \cite{hardy2001quantum}, $\VS$ is sufficient to make unforgeable money.
\end{cor}
}
\section{Conclusion}

In this paper we have considered Wiesner's quantum money scheme in the generalised probabilistic theory framework.
We first defined the class of GPTs in which there is potential for unforgeable money, that is, those satisfying the $\WNC$ assumption.
We then demonstrate that under an assumption of Verification Sharpness, this is equivalent to the inability to broadcast arbitrary causal states --- a general feature of non-classical GPTs, for example, those in \cite{barnum2007generalized}.

To obtain meaningful security however, we require that the probability of successfully counterfeiting can be made arbitrarily close to zero, that is, the assumption of $\SNC$.
Demonstrating that this is indeed possible for arbitrary GPTs satisfying $\WNC$ is the main result of this paper.
That is, we have a dichotomy: for any theory, either practical security is possible or perfect counterfeiting is possible.

It would be interesting to see if this work could be extended in a way similar to the quantum schemes that have been developed in recent years. For example, is it possible to have a purely classical verification scheme in the GPT setting?
Would it be possible to have a store merchant be able to verify the money without the need of involving the bank?
Both of these are possible in the quantum setting, see \cite{gavinsky2012quantum} and \cite{aaronson2012quantum}, respectively.
These scenarios are interesting, but add to the complexity of the problem considerably. As this work shows that GPT money is indeed physically possible in many theories, it paves the way to look at these more elaborate scenarios which are more convenient for the bank.

\section*{Acknowledgments}
This research was supported in part by Perimeter Institute for Theoretical Physics. Research at Perimeter Institute is supported by the Government of Canada through the Department of Innovation, Science and Economic Development Canada and by the Province of Ontario through the Ministry of Research, Innovation and Science.

\bibliographystyle{plainnat}
\bibliography{bibliography}

\appendix


\section{Proof of Lemma~\ref{lem:WNC}}
\label{appproof}

\begin{proof}
First let us assume perfect counterfeiting, i.e., $\alpha_A=1$.
{Since the value of $\alpha_A$ is attained (see Remark~\ref{rem:alpha}), we know that there exists a subcausal $\chi$ such that}
\begin{equation}
\InputIfFileExists{Diagrams/verifSharp3.tikz}{}{\input{./figures/Diagrams/verifSharp3.tikz}}\ =\ 1 \encomma
\end{equation}
for all $i$. $\VS$ therefore implies that for all $i$, we have
\begin{equation}
\InputIfFileExists{Diagrams/verifSharp4.tikz}{}{\input{./figures/Diagrams/verifSharp4.tikz}} \enperiod
\end{equation}
As each $s_i$ is causal, this implies that for all $i$,
\begin{equation}
\InputIfFileExists{Diagrams/verifSharp5.tikz}{}{\input{./figures/Diagrams/verifSharp5.tikz}} \enperiod
\end{equation}
Using $\VS$ again, we have
\begin{equation}
\InputIfFileExists{Diagrams/verifSharp6.tikz}{}{\input{./figures/Diagrams/verifSharp6.tikz}}
\end{equation}
for all $i$, which is precisely what it means for $\chi$ to be able to broadcast all of the states $s_i$.

For the other direction, let us assume there is a channel $B$ which broadcasts the states $s_i$:
\begin{equation}
\InputIfFileExists{Diagrams/broadcast1.tikz}{}{\input{./figures/Diagrams/broadcast1.tikz}} \enperiod
\end{equation}
Therefore we have, for all $i$:
\begin{equation} \label{eq34}
\InputIfFileExists{Diagrams/broadcast2.tikz}{}{\input{./figures/Diagrams/broadcast2.tikz}} \enperiod
\end{equation}
As we know $e_i$ are subcausal, we can decompose the discarding map as $e_i+E_i$ where each $E_i$ is a subcausal effect.
This gives, for all $i$:
\begin{equation}\label{eq34NEW}
\InputIfFileExists{Diagrams/broadcast3.tikz}{}{\input{./figures/Diagrams/broadcast3.tikz}} \enperiod
\end{equation}
We use the subcausality of $B$ (and the causality of each $s_i$) to write
\begin{equation} \label{pics}
\InputIfFileExists{Diagrams/broadcast4.tikz}{}{\input{./figures/Diagrams/broadcast4.tikz}}
\end{equation}
using the same decomposition of the discarding map as above.
Combining Eqs.~\eqref{eq34NEW} and \eqref{pics}, we have
\begin{equation}
\InputIfFileExists{Diagrams/broadcast5.tikz}{}{\input{./figures/Diagrams/broadcast5.tikz}}
\end{equation}
and
\begin{equation}
\InputIfFileExists{Diagrams/broadcast6.tikz}{}{\input{./figures/Diagrams/broadcast6.tikz}}\ =\ 1
\end{equation}
for all $i$, {since each number is nonnegative}. This implies that perfect counterfeiting is possible.
\end{proof}

\section{Proof of Lemma~\ref{usefulLemma}
}
\label{AppA}
The following well-known fact (see, e.g., \cite{BV}) characterises the interior of the dual cone.

\begin{fact}
For a \emph{closed} convex cone $K$, we have
\begin{equation}
\mathrm{int}(K^*) = \{ y : \inner{y}{x} > 0, \text{ for all } x \in K \setminus \{ 0 \} \}.
\end{equation}
\end{fact}

{We are now ready to prove  Lemma~\ref{usefulLemma}.}

\proof
We want to show that, for $\chi \in K_A^{AA}$,
\begin{equation}
\InputIfFileExists{Diagrams/interiorPoint1.tikz}{}{\input{./figures/Diagrams/interiorPoint1.tikz}}\ =\ 0 \quad \implies \quad %
\InputIfFileExists{Diagrams/interiorPoint2.tikz}{}{\input{./figures/Diagrams/interiorPoint2.tikz}}\ = \ 0.
\end{equation}
To do so, however, {recall} the notion of \emph{tomography} which is how equality is defined for processes in GPTs. We have two processes $f,g:A\to B$ are equal if and only if
\begin{equation}
\InputIfFileExists{Diagrams/tomographyf.tikz}{}{\input{./figures/Diagrams/tomographyf.tikz}}\ = \ %
\InputIfFileExists{Diagrams/tomographyg.tikz}{}{\input{./figures/Diagrams/tomographyg.tikz}} \quad \forall C, s\in K^{AC}, e\in K_{BC}.
\end{equation}
Therefore we need to check that
\begin{equation}
\InputIfFileExists{Diagrams/interiorPoint3.tikz}{}{\input{./figures/Diagrams/interiorPoint3.tikz}}\ =\ 0 \quad  \forall C, s\in K^{AC}, e\in K_{AAC}.
\end{equation}
To see this first note that as $y$ is interior, any other state, $a\in K^A$ belongs to some convex decomposition of $y$. Therefore we have
\begin{equation}
\InputIfFileExists{Diagrams/interiorPoint1.tikz}{}{\input{./figures/Diagrams/interiorPoint1.tikz}}\ =\ 0 \quad \implies \quad %
\InputIfFileExists{Diagrams/interiorPoint4.tikz}{}{\input{./figures/Diagrams/interiorPoint4.tikz}}\ = \ 0 \quad \forall a\in K^A.
\end{equation}
In particular, we can take $a$ to be the `marginal' of some bipartite state, $s$, i.e.
\begin{equation}
\InputIfFileExists{Diagrams/interiorPoint5.tikz}{}{\input{./figures/Diagrams/interiorPoint5.tikz}}\ = \ 0.
\end{equation}
This holds for any state $s$ and any system $C$ so we can rewrite this as
\begin{equation}
\InputIfFileExists{Diagrams/interiorPoint6.tikz}{}{\input{./figures/Diagrams/interiorPoint6.tikz}}\ = \ 0 \quad \forall C, s\in K^{AC}.
\end{equation}
Note that  the composite of the three discarding effects is the discarding effect for the composite system. This discarding effect is in the interior of $K_{AAC}$, so any effect $e\in K_{AAC}$ can appear in some convex decomposition of the discarding effect. Therefore we have
\begin{equation}
\InputIfFileExists{Diagrams/interiorPoint3.tikz}{}{\input{./figures/Diagrams/interiorPoint3.tikz}}\ = \ 0 \quad \forall C, s\in K^{AC}, e\in K_{AAC}.
\end{equation}
This concludes our proof as this is exactly the condition needed for tomography to show that $\chi=0$.
\endproof

\end{document}